\documentclass[preprint,3p,number]{elsarticle}

\usepackage{amssymb}
\usepackage{url}
\usepackage{amsmath}
\usepackage{amsthm}
\usepackage{alltt}
\usepackage{mathtools}
\usepackage{dsfont}
\usepackage[T1]{fontenc}
\usepackage{fancybox} 

\newtheorem{theorem}{Theorem}
\newtheorem{lemma}[theorem]{Lemma}
\newtheorem{remark}[theorem]{Remark}

\newtheorem{definition}[theorem]{Definition}

\newcommand\mfbox[1]{\,\,\fbox{#1}\,\,}
\newcommand\K{$\mathds K$}
\newcommand\exec{\textnormal{\,\guillemotright\,}}
\newcommand\ifthenelse[3]{\textnormal{\bf if~}#1\textnormal{\bf~then~}#2\textnormal{\bf~else~}#3\textnormal{\bf~end}}
\newcommand\thenelse[2]{\textnormal{\bf then~}#1\textnormal{\bf~else~}#2\textnormal{\bf~end}}
\newcommand\loopend[1]{\textnormal{\bf loop~}#1\textnormal{\bf~end}}
\newcommand\create[1]{\textnormal{\bf create~}#1}
\newcommand\forget[1]{\textnormal{\bf forget~}#1}
\newcommand\call[1]{\textnormal{\bf call}\,#1}
\newcommand\assgn{\,:\!=\,}
\newcommand\lasso[2]{\textnormal{lasso(}#1,#2\textnormal{)}}
\newcommand\reg[2]{\textnormal{reg(}#1,#2\textnormal{)}}

\newcommand\cellS[2]{\small \langle\,\, #1 \,\,\rangle_{\textnormal{#2}}}
\newcommand\cellSL[2]{\small \langle\,\, #1  \,\,\ldots\rangle_{\textnormal{#2}}}
\newcommand\cellSR[2]{\small \langle\ldots\,\,  #1 \,\,\rangle_{\textnormal{#2}}}

\newcommand\cellU[3]
{
\small
\begin{tabular}{r@{}c@{}l}
$\langle\,\,$ & $#1$ & $\,\,\rangle_{\textnormal{#3}}$\\
\cline{2-2}
& $#2$ &\\
\end{tabular}
}
\newcommand\cellUL[3]
{
\small
\begin{tabular}{r@{}c@{}l}
$\langle\,\,$ & $#1$ & $ \,\,\ldots\rangle_{\textnormal{#3}}$\\
\cline{2-2}
& $#2$ &\\
\end{tabular}
}
\newcommand\cellUR[3]
{
\small
\begin{tabular}{r@{}c@{}l}
$\langle\ldots\,\, $ & $#1$ & $\,\,\rangle_{\textnormal{#3}}$\\
\cline{2-2}
& $#2$ &\\
\end{tabular}
}

\mathcode`:="003A  
\mathcode`;="003B  
\mathcode`?="003F  
\mathcode`|="026A  
\mathcode`<="4268  
\mathcode`>="5269  

\mathchardef\ls="213C    
\mathchardef\gr="213E    

\newenvironment{todo}{\bigskip\hrule\medskip\noindent}{\medskip\hrule\bigskip}

\journal{Science of Computer Programming}

\begin{document}

\begin{frontmatter}



\title{On the Verification of SCOOP Programs}


\author[ETH]{Georgiana Caltais}
\ead{georgiana.caltais@inf.ethz.ch}
\author[ETH,RU]{Bertrand Meyer}
\ead{bertrand.meyer@inf.ethz.ch}

\address[ETH]{Department of Computer Science, ETH Z\"urich, Switzerland}
\address[RU]{Software Engineering Lab, Innopolis University, Russia}

\begin{abstract}
In this paper we focus on the development of a toolbox for the verification of programs in the context of SCOOP -- an elegant concurrency model, recently formalized based on Rewriting Logic (RL) and Maude.
SCOOP is implemented in Eiffel and its 
applicability is demonstrated also from a practical perspective, in the area of robotics programming.
Our contribution consists in devising and integrating an alias analyzer and a Coffman deadlock detector under the roof of the same RL-based semantic framework of SCOOP.
This enables using the Maude rewriting engine and its LTL model-checker ``for free'', in order to perform the analyses of interest.
We discuss the limitations of our approach for model-checking deadlocks and provide solutions to the state explosion problem. The latter is mainly caused by the size of the SCOOP formalization which incorporates all the aspects of a real concurrency model.
On the aliasing side, we propose an extension of a previously introduced alias calculus based on program expressions, to the setting of unbounded program executions such as infinite loops and recursive calls.
Moreover, we devise a corresponding executable specification easily implementable on top of the SCOOP formalization.
An important property of our extension is that, in non-concurrent settings, the corresponding alias expressions can be over-approximated in terms of a notion of regular expressions. This further enables us to derive an algorithm that always stops and provides a sound over-approximation of the ``may aliasing'' information, where soundness stands for the lack of false negatives.
\end{abstract}


\end{frontmatter}


\section{Introduction}\label{intro}

In light of the widespread deployment and complexity of concurrent systems, the development of corresponding frameworks for rigorous design and analysis has been a great challenge.
Along this research direction, the focus can be two-fold. On the one hand, the interest might be on formalizing and reasoning about a concurrency model, its characteristic concepts and synchronization mechanisms, for instance. On the other hand, of equal importance, the efforts might be directed towards checking the behavior of concurrent applications and their properties. These two research areas are closely related: it is often the case that, for example, analysis tools for concurrent applications depend on the underlying concurrency model. Hence, the development of a \emph{unifying framework for the design and analysis of both the model and its applications} is of interest.

In this paper we are targeting SCOOP~\cite{DBLP:conf/acsd/MorandiSNM13}, a simple object-oriented programming model for concurrency. Two main characteristics make SCOOP simple: 1) just one keyword programmers have to learn and use in order to enable concurrent executions, and 2) the burden of orchestrating concurrent executions is handled within the model, therefore reducing the risk of correctness issues. The reference implementation is Eiffel~\cite{DBLP:books/ph/Meyer91}, but implementations have also been built on top of languages such as Java. The success of SCOOP is demonstrated not only from a research perspective, but also from a practical perspective, with applications appearing, for instance, in the area of robotics programming~\cite{DBLP:conf/iros/RusakovSM14}.

The basis of a framework for the design and analysis of the SCOOP model has already been set. In this respect, we refer to the recent formalization of SCOOP in~\cite{DBLP:conf/acsd/MorandiSNM13} based on Rewriting Logic (RL)~\cite{DBLP:conf/fct/MeseguerR11}, which is ``executable'' and straightforwardly implementable in the programming language Maude~\cite{DBLP:conf/maude/2007}. In~\cite{DBLP:conf/maude/2007} these capabilities have been successfully exploited in order to reason on the original SCOOP model and to identify a number of design flaws.

Moreover, an executable semantics can be exploited in order to formalize and ``run'' analysis tools for SCOOP programs as well. This facilitates the extension of the aforementioned SCOOP formalization to the level of a unifying executable semantic framework for the design and analysis of both the model and its concurrent applications.
In this paper we focus on the \emph{development of a RL-based toolbox for the analysis of SCOOP programs} on top of the formalization in~\cite{DBLP:conf/acsd/MorandiSNM13}. We are interested in constructing an {alias analyzer} and a {deadlock detector}.

\emph{Alias analysis} has been an interesting research direction for the verification and optimization of programs.
One of the challenges along this line of research has been the undecidability of determining whether two expressions in a program \emph{may} reference the same object. A rich suite of  approaches aiming at providing a satisfactory balance between scalability and precision has already been developed in this regard. Examples include:
 (i) intra-procedural frameworks~\cite{Landi:1991:PAP:99583.99599,Landi:1992:USA:161494.161501} that handle isolated functions only, and their inter-procedural counterparts~\cite{Landi:1992:USA:161494.161501,Myers:1981:PID:567532.567556,Hind:1999:IPA:325478.325519} that consider the interactions between function calls;
(ii) {type-based} techniques~\cite{Diwan:1998:TAA:277652.277670};
(iii) flow-based techniques~\cite{Burke-flow-ins,Choi:1993:EFI:158511.158639} that establish aliases depending on the control-flow information of a procedure;
(iv) context-(in)sensitive approaches~\cite{Emami:1994:CIP:178243.178264,Wilson:1995:ECP:207110.207111} that depend on whether the calling context of a function is taken into account or not;
(v) field-(in)sensitive approaches~\cite{Mine:2006:FVA:1134650.1134659,Albert:2009:FVA:1693345.1693376} that depend on whether the individual fields of objects in a program are traced or not.

There is a huge literature on heap analysis for aliasing~\cite{DBLP:conf/paste/Hind01}, but hardly any paper that presents a calculus allowing the derivation of alias relations as the result of applying various instructions of a programming language. Hence, of particular interest for the work in this paper is the untyped, flow-sensitive, field sensitive, inter-procedural and context-sensitive calculus for {may aliasing}, introduced in~\cite{Meyer-aliasing-13}. The calculus covers most of the aspects of a modern object-oriented language, namely: object creation and deletion, conditionals, assignments, loops and (possibly recursive) function calls. The approach in~\cite{Meyer-aliasing-13} abstracts the aliasing information in terms of explicit access paths~\cite{DBLP:conf/pldi/LarusH88} referred to as \emph{alias expressions} straightforwardly computed in an equational fashion, based on the language constructs. As we shall see later on in this paper, the language-based expressions can be exploited in order to reason on ``may aliasing'' in a finite number of steps in non-concurrent settings and, moreover, can be easily incorporated in the semantic rules defining SCOOP in~\cite{DBLP:conf/acsd/MorandiSNM13}.

\emph{Deadlock} is one of the most serious problems in concurrent systems. It occurs when two or more executing threads are each waiting for the other to finish.
Along time, the complexity of the problem determined various approaches to combat deadlocks~\cite{Shih:1990:SDD:896944}. Examples include:
(i) deadlock prevention~\cite{Andrews:1982:ODP:800220.806694} which ensures that at least one of the deadlock conditions cannot hold, (ii) deadlock avoidance~\cite{Minoura:1982:DAR:322344.322351} that provides a priori information so that the system can predict and avoid deadlock situations, (iii) deadlock detection~\cite{chandy1982distributed,badal1983deadlock} that detects and recovers from a deadlock state.

Our focus is on \emph{deadlock detection} for SCOOP programs. We base our work on the fact that this type of analysis is in strict connection with the underlying model of interest. Consequently, as described in the corresponding subsequent sections, our approach consists in formalizing deadlocks in the context of the SCOOP concurrency model and enriching its semantics in~\cite{DBLP:conf/acsd/MorandiSNM13} with the equivalent operational-based definition of deadlocks.
This enables using the Maude rewriting capabilities ``for free'' in order to test SCOOP programs for deadlock.
Nevertheless, the more ambitious goal of  using the Maude LTL model-checker for deadlock detection is not straightforward. As discussed in more detail later on in this paper, verification of deadlocks was possible after reducing the SCOOP semantics in~\cite{DBLP:conf/acsd/MorandiSNM13} and abstracting it based on aliasing information, and modifying a series of implementation aspects (such as indexed-based parameterizations) that determined state explosion issues.


\newpage
\paragraph{Our contribution}{
This paper is an extended version of~\cite{DBLP:journals/corr/Caltais14} where we proposed:
\begin{enumerate}\itemsep1pt
\item[1.] a translation of the (finite) alias calculus in~\cite{Meyer-aliasing-13} to the setting of unbounded program executions such as infinite loops and recursive calls, together with a sound over-approximation technique based on  (finitely representable) ``regular alias expressions'' capturing unbounded executions in non-concurrent settings;
\item[2.] a RL-based specification of the extended calculus  suitable for integration within the SCOOP formalization in~\cite{DBLP:conf/acsd/MorandiSNM13} (for this purpose we chose the {\K} semantic framework as a RL-based formalism enabling compact and modular definitions);
\item[3.] an algorithm for ``may aliasing'' (exploiting the finiteness property in 1.) that always terminates in non-concurrent settings.
\end{enumerate}
The current work adds to 1.--3. above:
\begin{enumerate}\itemsep1pt
\item[4.] the full RL-based specification in 2. and the complete formal proofs showing the soundness of the over-approximating technique based on ``regular alias expressions'';
\item[5.] examples of exploiting the algorithm in 3. and its implementation on top of the SCOOP formalization in Maude~\cite{DBLP:conf/acsd/MorandiSNM13};
\item[6.] the formalization and integration of a deadlock detection mechanism on top of the SCOOP operational semantics~\cite{DBLP:conf/acsd/MorandiSNM13}, together with discussions on the limitations of our approach and associated workarounds.
\end{enumerate}
}

\paragraph{Paper structure}{
The paper is organized as follows. 
In Section~\ref{sec:SCOOP} we provide a brief overview of SCOOP.
In Section~\ref{sec:alias-calc} we introduce the extension of the alias calculus in~\cite{Meyer-aliasing-13} to unbounded executions. In Section~\ref{sec:implem-k} we provide the full RL-based executable specification of the calculus. The implementation in SCOOP and further applications are discussed in Section~\ref{sec:alias-SCOOP}.
Section~\ref{sec:coffman-deadlocks} is dedicated to deadlocking in SCOOP.
In Section~\ref{sec:discussion} we draw the conclusions, discuss some of the related works and provide pointers to future developments.
}

\section{Biref introduction to  SCOOP}
\label{sec:SCOOP}

As already stated, the purpose of the current work is the development of a toolbox for the analysis of SCOOP programs by exploiting the semantics proposed in~\cite{DBLP:conf/acsd/MorandiSNM13}.
SCOOP is particularly attractive due to its simplicity and elegance, as it allows the switch from sequential to concurrent programming in a rather straightforward fashion, by means of just one keyword, namely, \emph{separate}. Transparent to the user, the key notion in SCOOP is the processor, or \emph{handler} (that can be a CPU, or it can also be implemented in software, as a process or thread). Handlers are in charge of executing the routines of ``separate'' objects, in a concurrent fashion.

For an example, assume a processor $p$ that performs a call $o.f(a_1, a_2, \ldots)$ on an object $o$. If $o$ is declared as ``separate'', then $p$ sends a request for executing $f(a_1, a_2, \ldots)$ to $q$ -- the handler of $o$ (note that $p$ and $q$ can coincide). Meanwhile, $p$ can continue. Moreover, assume that $a_1, a_2, \ldots$ are of ``separate'' types. According to the SCOOP semantics, the application of the call  $f(\ldots)$ will \emph{wait} until it has been able to \emph{lock} all the separate objects associated to $a_1, a_2, \ldots$. This mechanism guarantees exclusive access to these objects. 
Given a processor $p$, by $W(p)$ we denote the set of processors $p$ \emph{waits} to release the resources $p$ needs for its asynchronous execution. Orthogonally, by $H(p)$ we represent the set of resources (more precisely, resource handlers that) $p$ already acquired.

The semantics of SCOOP in~\cite{DBLP:conf/acsd/MorandiSNM13} is defined over tuples of shape   
\begin{equation}
\label{eq:tuple}
\langle p_1 \,::\, St_{1} \mid \ldots \mid p_n \,::\, St_{n}, \sigma \rangle
\end{equation}
where, $p_i$ denotes a processor (for $i \in \{1, \ldots, n\}$), $St_i$ is the call stack of $p_i$ and $\sigma$ is the {\it state} of the system. States hold information about the {\it heap} (which is a mapping of references to objects) and the {\it store} (which includes formal arguments, local variables, {\it etc}.).
Processors communicate via \emph{channels}.

Roughly speaking, one could classify the operational rules formalizing SCOOP in~\cite{DBLP:conf/acsd/MorandiSNM13} in: a) \emph{language rules} that provide the semantics of language constructs such as ``{\bf if \ldots then \ldots else \ldots end}'' or ``{\bf until \ldots loop \ldots end}'', and b) \emph{control rules} implementing mechanisms such as locking or scheduling.

For an example in category a) above, consider the rules specifying ``{\bf if}'' instructions:
\begin{equation}
\label{eq:if-then-else}
\dfrac{\textnormal{a is fresh}}{
\begin{array}{cc}
\langle p \,::\,\ifthenelse{e}{St_1}{St_2}\,;\, St,\, \sigma \rangle \rightarrow\\
\langle p\,::\, \textnormal{eval}(a, e);\, \textnormal{wait}(a);\, provided~a.data~then~St_1~else~St_2;\,St,\, \sigma\rangle
\end{array}
}
\end{equation}
\begin{equation}
\label{eq:if-then}
\dfrac{.}{
\langle p\,::\, provided~true~then~St_1~else~St_2;\,St,\, \sigma\rangle \rightarrow
\langle p\,::\, St_1;\,St,\, \sigma\rangle
}
\end{equation}
\begin{equation}
\label{eq:if-else}
\dfrac{.}{
\langle p\,::\, provided~false~then~St_1~else~St_2;\,St,\, \sigma\rangle \rightarrow
\langle p\,::\, St_2;\,St,\, \sigma\rangle
}
\end{equation}

Intuitively, ``eval$(a, e)$'' evaluates $e$ and puts the result on a fresh channel $a$ and ``wait$(a)$'' enables processor $p$ to use the evaluation result stored in $a.data$. It is straightforward to see that, according to~(\ref{eq:if-then}), in case the condition $e$ is evaluated to $true$ then the ``{\bf if} branch'' $St_1$ is placed on top of the call stack of $p$. Otherwise, based on~(\ref{eq:if-else}), if $e$ is evaluated to $false$, the ``{\bf else} branch'' is executed.

As we shall see in Section~\ref{sec:alias-SCOOP}, an operational view on the alias calculus in~\cite{Meyer-aliasing-13} exploiting the instructions of a programming language will enable a straightforward implementation on top of the ``language rules'' of SCOOP.

For the case b) above we refer to the locking rule:
\begin{equation}
\label{eq:locks}
\dfrac{\forall q_i \in \{q_1,\ldots,q_m\}\,:\,\sigma.rq\_locked(q_i) = false}
{
\begin{array}{c}
\langle p\,::\,lock(\{q_1,\ldots\,q_m\});\,St,\, \sigma\rangle \rightarrow\\
\langle p\,::\,St,\, \sigma.lock\_rqs(p, \{q_1,\ldots\,q_m\})\rangle
\end{array}
}
\end{equation}
\noindent
stating that a processor $p$ can lock a set of handlers $\{q_1,\ldots,q_m\}$ by calling $lock\_rqs$ on the state $\sigma$ whenever none of the handlers $q_i$ has previously been acquired by other processors, {\it i.e.}, $\sigma.rq\_locked(q_i) = false$.

As it will become clear in Section~\ref{sec:coffman-deadlocks}, ``control rules'' pave the way to an immediate implementation of a corresponding ``deadlock rule'' on top of the Maude formalization of SCOOP in~\cite{DBLP:conf/acsd/MorandiSNM13}.

\section{The alias calculus}
\label{sec:alias-calc}

The calculus for \emph{may aliasing} introduced in~\cite{Meyer-aliasing-13} abstracts the aliasing information in terms of explicit access paths referred to as ``alias expressions''. Consider, for an example,the  case of a linked list. We write $x_i$ ($i \geq 0$) to represent node $i$ in the list, and use a setter to assign the next node of the list:
\begin{equation}
\label{eq:intor-ex}
\begin{array}{l}
\create{x_0}\\
\loopend{
\\\hspace{10pt}i\,:=i+1
\\\hspace{10pt}\create{x_i}
\\\hspace{10pt}x_i . set\_next(x_{i-1})\\\hspace{-4.5pt}
}
\end{array}
\end{equation}
The result of the execution of the code above can be intuitively depicted as the infinite sequence:
\[
x_0 \xleftarrow{next} x_1 \xleftarrow{next} \ldots x_{k-1} \xleftarrow{next} x_k \xleftarrow{next} x_{k+1} \ldots
\]
Hence, $x_0$ becomes aliased to $x_1.next$, $x_2.next.next$, $x_3.next.next.next$, so on and so on. In short, the set of associated alias expressions can be equivalently written as:
\begin{equation}
\label{eq:aliasing-example}
\{[x_i,\, x_{i+k}.next^k] \mid i \geq 0\land k \geq 1\}.
\end{equation}
The sources of imprecision introduced by the calculus in~\cite{Meyer-aliasing-13} are limited to ignoring tests in conditionals, and to ``cutting at length $L$'' for the case of possibly infinite alias relation corresponding to unbounded executions as in~(\ref{eq:intor-ex}). The cutting technique considers sequences longer than a given length $L$ as aliased to all expressions.

In this section we define an extension of the calculus in~\cite{Meyer-aliasing-13}, to unbounded program executions. Moreover, based on the idea behind the \emph{pumping lemma for regular languages}~\cite{Rabin:1959:FAD:1661907.1661909}, we devise a corresponding sound over-approximation of ``may aliasing'' in terms of regular expressions, applicable in sequential contexts. This paves the way to developing an algorithm for the aliasing problem, as presented in Section~\ref{sec:implem-k}, in the formal setting of the {\K} semantic framework~\cite{DBLP:journals/jlp/RosuS10}. Note that {\K} is used more as a notational convention, as its operational flavor enables a straightforward integration within the SCOOP formalization in~\cite{DBLP:conf/acsd/MorandiSNM13}.

\paragraph{Brief overview of the alias calculus}{
We proceed by recalling the notion of \emph{alias relation} and a series of associated notations and basic operations, as introduced in~\cite{Meyer-aliasing-13}.

We call an \emph{expression} a (possibly infinite) path of shape $x.y.z.\,\ldots$, where $x$ is a local variable, class attribute or {\it{Current}}, and $y, z, \ldots$ are attributes. Here, {\it{Current}}, also known as {\it{this}} or {\it{self}}, stands for the current object.
For an arbitrary alias expression $e$, it holds that $e . {\it Current} = {\it Current} . e = e$. 
Let $E$ represent the set of all expressions of a program. An \emph{alias relation} is a symmetric and irreflexive binary relation over $E \times E$.

Given an alias relation $r$ and an expression $e$, we define
\[
r/e = \{e\} \cup \{x:\, E \mid [x,e] \in r\}
\]
denoting the set consisting of all elements in $r$ which are aliased to $e$, plus $e$ itself.

Let $x$ be an expression; we write $r\,-\,x$ to represent $r$ without the pairs with one element of shape $x.e$.

We say that an alias relation is \emph{dot complete} whenever for any $t, u,v$ and $a$ it holds that if $[t,u]$ and $[t.a,v]$ are alias pairs, then $[u.a, v]$ is an alias pair and, moreover, if $a$ is in the domain of $t$, then $[t.a, u.a]$ is an alias pair.
By the ``domain of $t$'' we refer to a method or a field in the class corresponding to the object referred by the expression associated to $t$. For instance, given a class NODE with a field $next$ of type NODE, and a NODE object $x$, we say that $next$ is in the domain of $t \,=\, x.next.next$.
For the sake of brevity, we write {\it dot-complete}$(r)$ for the closure under dot-completeness of a relation $r$.

The notation $r[x=u]$ represents the relation $r$ augmented with pairs $[x,y]$ and made dot complete, where $y$ is an element of $u$.
}

\subsection{Extension to unbounded executions}
\label{sec:ext-inf-expr}

We further introduce an extension of the alias calculus in~\cite{Meyer-aliasing-13} to infinite alias relations corresponding to unbounded executions such as infinite loops or recursive calls. The main difference in our approach is reflected by the definition of loops, which now complies to the usual fixed-point denotational semantics.

The alias calculus is defined by a set of axioms ``describing'' how the execution a program affects the aliasing between expressions. As in~\cite{Meyer-aliasing-13}, the calculus ignores tests in conditionals and loops. The \emph{program instructions} are defined as follows:
\begin{equation}
\label{eq:BNF-control-struct}
\begin{array}{r c l}
p & ::= & p \,;\, p \mid \thenelse{p}{p} \mid\\
& & \create{x} \mid \forget{x} \mid t \assgn s \mid\\
&& \loopend{p} \mid \call{f(l)} \mid x.\call{f(l)}.
\end{array}
\end{equation}
In short, we write
$
r \exec p
$
to represent the alias information obtained by executing $p$ when starting with the initial alias relation $r$.

The axiom for sequential composition is defined in the obvious way:
\begin{equation}
\label{eq:def-seq-comp}
r \exec (p\,;\,q) = (r \exec p) \exec q.
\end{equation}

Conditionals are handled by considering the union of the alias pairs resulted from the execution of the instructions corresponding to each of the two branches, when starting with the same initial relation:
\begin{equation}
\label{eq:def-then-else}
r \exec (\thenelse{p}{q})  = r \exec p\,\, \cup\,\, r \exec q.
\end{equation}

As previously mentioned,  we define $r \exec \loopend{p}$ according to its informal semantics : ``execute $p$ repeatedly any number of times, including zero''. The corresponding rule is:
\begin{equation}
\label{eq:def-loop}
r \exec (\loopend{p}) = \bigcup_{n \in \mathds{N}} (r \exec p^n)
\end{equation}
where $\cup$ stands for the union of alias relations, as above.
This way, our calculus is extended to infinite alias relations. 
This is the main difference with the approach in~\cite{Meyer-aliasing-13} that proposes a ``cutting'' technique restricting the model to a maximum length $L$.  In~\cite{Meyer-aliasing-13}, sequences longer than $L$ are considered as aliased to all expressions. Orthogonally, for sequential settings, we provide finite representations of infinite alias relations based on over-approximating regular expressions, as we shall see in Section~\ref{sec:sound-over-approx}. 

Both the creation and the deletion of an object $x$ eliminate from the current alias relation all the pairs having one element prefixed by $x$:
\begin{equation}
\label{eq:def-create-delete}
\begin{array}{rcl}
r \exec (\create{x}) & = & r - x\\
r \exec (\forget{x}) & = & r - x.
\end{array}
\end{equation}

The (qualified) function calls comply to their initial definitions in~\cite{Meyer-aliasing-13}: 
\begin{equation}
\label{eq:def-qualified-call}
\begin{array}{rcl}
r \exec (\call{f(l)}) & = & (r[f^\bullet:l])\exec \mid f \mid\\
r \exec (x.\call{f(l)}) & = & x.((x'.r) \exec \call{f(x'.l)}).
\end{array}
\end{equation}
Here $f^\bullet$ and $\mid f \mid$ stand for the formal argument list and the body of $f$, respectively, whereas $r[u:v]$ is the relation $r$ in which every element of the list $v$ is replaced by its counterpart in $u$. Intuitively, the negative variable $x'$ is meant to transpose the context of the qualified call to the context of the caller. Note that ``$.$'' ({\it i.e.}, the constructor for alias expressions) is generalized to distribute over lists and relations:
$x.[a,b,\ldots] = [x.a, x.b, \ldots]$.

For an example, consider a class $C$ in an OO-language, and an associated procedure $f$ that assigns a local variable $y$, defined as: $f(x) \,\{ \,\, y \assgn x \,\, \}$.
Then, for instance, the aliasing for $a.\call{f(a)}$ computes as follows:
\[
\begin{array}{rc}
\emptyset \,\,\exec\,\, a.\call{f(a)} & = \\
a.(a'.\emptyset \,\,\exec\,\, y\,:=a'.a) & =\\
a.(\emptyset \,\,\exec\,\, y\,:=\textnormal{\it Current}) & =\\
\textnormal{\it dot-complete}(\{[a.y, a]\}).
\end{array}
\]

Recursive function calls can lead to infinite alias relations. In sequential settings, as for the case of loops, the mechanism exploiting sound regular over-approximations in order to derive finite representations of such relations is presented in the subsequent sections.

The axiom for assignment is as well in accordance with its original counterpart in~\cite{Meyer-aliasing-13}:
\begin{equation}
\label{eq:def-assign}
\begin{array}{rcl}
r \exec (t \assgn s) & = & {\textnormal{\bf given~}} r_{1} = r[ot = t]\\
&& {\textnormal{\bf then~}} (r_{1} - t)[t\,=\,(r_1 \slash s \,-\,t)] - ot  {\textnormal{\bf~end}}\\
\end{array}
\end{equation}
where $ot$ is a fresh variable (that stands for ``old $t$'').
Intuitively, the aliasing information w.r.t. the initial value of $t$ is ``saved'' by associating $t$ and $ot$ in $r$ and closing the new relation under dot-completeness, in $r_1$. Then, the initial $t$ is ``forgotten'' by computing $r_1 - t$ and the new aliasing information is added in a consistent way. Namely, we add all pairs $(t, s')$, where $s'$ ranges over {$r_1 \slash s \,-\, t$} representing all expressions already aliased with $s$ in $r_1$, including $s$ itself, but without $t$. Recall that alias relations are not reflexive, thus by eliminating $t$ we make sure we do not include pairs of shape $[t, t]$. Then, we consider again the closure under dot-completeness and forget the aliasing information w.r.t. the initial value of $t$, by removing $ot$.

\begin{remark}
It is worth discussing the reason behind \emph{not} considering transitive alias relations.
Assume the following program: 
\[
\thenelse{x \assgn y}{y \assgn z}
\]
Based on the equations~(\ref{eq:def-then-else}) and~(\ref{eq:def-assign}) handling conditionals and assignments, respectively, the calculus correctly identifies the alias set: $\{[x, y], [y, z]\}$. Including $[x, z]$ would be semantically equivalent to the execution of the two branches in the conditional at the same time, which is not what we want.
\end{remark}


\subsection{A sound over-approximation}
\label{sec:sound-over-approx}

\label{rm:unfolding}
In a sequential setting, the challenge of computing the alias information in the context of (infinite) loops and recursive calls reduces to evaluating their corresponding ``unfoldings'', captured by expressions of shape
\[
r \exec p^{\omega},
\]
with $\omega$ ranging over naturals plus infinity, r an (initial) alias relation ($r = \emptyset$), and $p$ a \emph{basic control block} defined by:
\begin{equation}
\label{eq:BNF-basic-struct}
\begin{array}{r c l}
p & ::= & p \,;\, p \mid \thenelse{p}{p} \mid\\
& & \create{x} \mid \forget{x} \mid\\
&& t \assgn s.
\end{array}
\end{equation}
The value $r \exec p^{\omega}$ refers to the alias relation obtained by recursively executing the control block $p$, and it is calculated in the expected way:
\[
\begin{array}{rcl}
r \exec p^{0} & = & r\\
r \exec p^{k+1} & = & (r \exec p^{k}) \exec p.
\end{array}
\]

Consider again the code in~(\ref{eq:intor-ex}):
\begin{equation}
\notag
\begin{array}{l}
\create{x_0}\\
\loopend{
\\\hspace{10pt}i\,:=i+1
\\\hspace{10pt}\create{x_i}
\\\hspace{10pt}x_i . set\_next(x_{i-1})\\\hspace{-4.5pt}
}
\end{array}
\end{equation}
Its execution generates an alias relation
including an infinite number of pairs of shape: 
\begin{equation}
\label{eq:inf-rel-next}
[x_i, x_{i+1}.next],\, [x_i, x_{i+2}.next.next],\, [x_i, x_{i+3}.next.next.next] \ldots~~.
\end{equation}
A similar reasoning does not hold for concurrent applications, where process interaction is not ``regular''.

In what follows we provide a way to compute finite representations of infinite alias relations in sequential settings.
The key observation is that alias expressions corresponding to unbounded program executions grow in a regular fashion. See, for instance, the aliases in~(\ref{eq:inf-rel-next}), which are pairs of type $[x_i, x_{i+k}.next^{k \geq 1}]$.

Regular expressions are defined similarly to the regular languages over an alphabet.
We say that an expression is \emph{regular} if it is a local variable, class attribute or {\it{Current}}. Moreover, the concatenation $e_1\,.\,e_2$ of two regular expressions $e_1$ and $e_2$ is also regular. Given a regular alias expression $e$, the expression $e^*$ is also regular; here $(-)^*$ denotes the Kleene star~\cite{Kleene56}. We call an alias relation \emph{regular} if it consists of pairs of regular expressions.

\begin{lemma}
\label{lm:sequential-regularity}
Assume $p$ a program built according to the rules in~(\ref{eq:BNF-control-struct}).
Then, in a sequential setting, the relation $\emptyset \exec p$ is regular.
\end{lemma}

In order to prove Lemma~\ref{lm:sequential-regularity}, we proceed by demonstrating a series of intermediate results.

\begin{remark}
\label{rm:op-preserve-reg}
We first observe that the operations $r \slash s$, $r - x$, dot-completeness and $r[x = u]$ introduced in Section~\ref{sec:alias-calc} preserve the regularity of an alias relation $r$. 
\end{remark}

Then, we define a notion of \emph{finite execution} control blocks:
\begin{equation}
\label{eq:BNF-basic-struct-finite}
\begin{array}{r c l}
p & ::= & \create{x} \mid \forget{x} \mid t \assgn s\mid \\
& & p \,;\, p \mid \thenelse{p}{p} \mid\\
& & \call{f(l)} \mid x . \call{f(l)}
\end{array}
\end{equation}
where $f$ stands for a non-recursive function.

It is easy to see that the execution of control blocks as in~(\ref{eq:BNF-basic-struct-finite}) preserve the regularity of alias relations as well.

\begin{lemma}
\label{lm:reg-fin-exec}
For all regular alias relations $r$ and $p$ a finite-execution control block, in a sequential setting, it holds that $r \exec p$ is also regular.
\end{lemma}
\begin{proof}
The proof follows immediately, by induction on the structure of $p$ and Remark~\ref{rm:op-preserve-reg}.
Base cases are: $\create{x}$, $\forget{x}$ and $t \assgn s$. For function calls, the result is a consequence of their corresponding unfolding, based on the definitions in~(\ref{eq:def-qualified-call}).
\end{proof}

\begin{remark}
\label{rm:rec-calls-loop}
With respect to may aliasing, recursive calls can be handled via loops.
Consider, for instance the recursive function
\[
f(x)\,\, \{\, B_1;\,\, f(y);\,\, B_2 \,\}
\]
where $B_1$ and $B_2$ are instruction blocks built as in~(\ref{eq:BNF-control-struct}). It is intuitive to see that computing the may aliases resulted from the execution of $f(x)$ reduces executing unfoldings of shape:
\[
\loopend{B_1};\,\,\loopend{B_2}.
\] 
\end{remark}

Moreover, unbounded program executions also preserve regularity.
\begin{lemma}
\label{lm:reg-infin-exec}
For all regular alias relations $r$ and $p$ a control block that can execute unboundedly, in a sequential setting, it holds that $r \exec p$ is also regular.
\end{lemma}
\begin{proof}
The proof follows by induction on the number of nested loops in $p$ and Remark~\ref{rm:rec-calls-loop}.
\end{proof}

At this point, the result in Lemma~\ref{lm:sequential-regularity} follows immediately by Lemma~\ref{lm:reg-fin-exec} and Lemma~\ref{lm:reg-infin-exec}.

Inspired by  the idea behind the \emph{pumping lemma for regular languages}~\cite{Rabin:1959:FAD:1661907.1661909}, we define a \emph{lasso} property for alias relations, which identifies the repetitive patterns within the structure of the corresponding alias expressions.
The intuition is that such patterns will occur for an infinite number of times due to the execution of loops or recursive function calls.
Then, we supply sound over-approximations of ``lasso'' relations, based on regular alias expressions.

In the context of alias relations, we say that the lasso property is satisfied by $r$ and $r'$ whenever the following two conditions hold: (1) $r$ behaves like a \emph{lasso base} of $r'$. Namely, all the pairs $[e_1, e_2] \in r$ are used to generate  elements $[e'_1, e'_2] \in r'$, by repeating tails of prefixes of $e_1$ and $e_2$, respectively, and (2) $r'$ is a \emph{lasso extension} of $r$. Namely, all the pairs in $r'$ are generated from elements of $r$ by repeating tails of their prefixes.
For example, if $e_1$ above is an expression of shape $x.y.z.w$, then $e'_1$ can be $x.y.y.z.w$ if we consider the tail $y$ of the prefix $x.y$, or $x.y.z.y.z.w$ if we take the tail $y.z$ of the prefix $x.y.z$.

Formally, consider $r$ and $r'$ two alias relations, and $x_i, y_i$ and $z_i$ a set of (possibly empty) expressions, for $i \in \{1,2\}$. Then:
\begin{equation}
\label{eq:def-lasso}
\lasso{r}{r'} = 
([x_1 y_1 z_1, x_2 y_2 z_2] \in r \textnormal{~~iff~~} [x_1 y_1 y_1 z_1, x_2 y_2 y_2 z_2] \in r').
\end{equation}
For the simplicity of notation we sometimes omit the dot-separators between expressions. For instance, we write $x\,y\,z$ in lieu of $x.y.z$.

Assuming a lasso over $r$ and $r'$, we compute a relation consisting of regular expressions over-approximating $r$ and $r'$ as:
\begin{equation}
\label{eq:def-reg}
\begin{array}{rcl}
\reg{r}{r'} & = & \{[x_1 y_1^* z_1, x_2 y_2^* z_2] \,\mid\\
&& \,\,\, [x_1 y_1 z_1, x_2 y_2 z_2] \in r\,\land\\
&& \,\,\, [x_1 y_1 y_1 z_1, x_2 y_2 y_2 z_2] \in r'\} 
\end{array}
\end{equation}
where $x_i, y_i$ and $z_i$ are possibly empty expressions, for $i \in \{1,2\}$. As previously indicated, the over-approximation is sound w.r.t. the repeated application of a basic control block as in~(\ref{eq:BNF-basic-struct}), in the way that it does not introduce any false negatives:

\begin{lemma}
\label{lm:reg-expr}
Consider $r$ and $r'$ two alias relations, and $p$ a basic control block in a sequential setting. If $r\exec p = r'$ and $\lasso{r}{r'} = true$, then the following holds for all $n\geq 1$:
\[
r \exec p^{n} \in \reg{r}{r'}.
\]
\end{lemma}
\begin{proof}
We proceed by induction on $n$.
\begin{itemize}
\item {\it Base case}: $n=1$. By hypothesis it holds that $\lasso{r}{r'} = true$. Hence, according to the definition of $\lasso{-}{-}$ in~(\ref{eq:def-lasso}), there exists a one-to-one correspondence of the shape
\[
[x_1 y_1 z_1, x_2 y_2 z_2] \in r \textnormal{~~iff~~} [x_1 y_1 y_1 z_1, x_2 y_2 y_2 z_2] \in r'
\]
between the elements of $r$ and $r'$, respectively.

Consequently, by the definition of $\reg{-}{-}$ in~(\ref{eq:def-reg}), it is easy to see that 
\[r' \in \reg{r}{r'}.\]

\item {\it Induction step.} Fix a natural number $n$ and suppose that
\begin{equation}
\label{eq:ind-step}
r \exec p^k \in \reg{r}{r'}
\end{equation}
for all $k \in \{1, \ldots, n\}$. We want to prove that~(\ref{eq:ind-step}) holds also for $k = n+1$.

We continue by ``reductio ad absurdum''. Consider 
\[
\overline{r} = r \exec p^n \in \reg{r}{r'},
\]
and assume that
\begin{equation}
\label{eq:red-absurdum}
\overline{r} \exec p \not \in \reg{r}{r'}
\end{equation}
Clearly, the execution of $p$ when starting with $\overline{r}$ identifies an alias pair which is not in $\reg{r}{r'}$. Given that $p$ is a basic control block as in~(\ref{eq:BNF-basic-struct}), and based on the corresponding definitions in~(\ref{eq:def-seq-comp})--(\ref{eq:def-assign}), it is not difficult to observe that the regular structure of the alias information can only be broken via a new added pair $(t, s')$ associated to an assignment $t \assgn s$ within $p$.

Let $p = C[t \assgn s]$, where $C$ is a context built according to~(\ref{eq:BNF-basic-struct}), and $t \assgn s$ is the upper-most assignment instruction in the syntactic tree associated to $p$, that introduces a pair $[t,s']$ which is not in $\reg{r}{r'}$. Assume that $\tilde{r}$ is the intermediate alias relation obtained by reducing $\overline{r} \exec C[t\assgn s]$ according to the equations~(\ref{eq:def-seq-comp})--(\ref{eq:def-assign}), before the application of the assignment axiom corresponding to $t\assgn s$.

Note that $t \assgn s$ was executed at least once before, as $n \geq 1$, and observe that $\tilde{r} \in \reg{r}{r'}$. Hence, we identify two situations in the context of the aforementioned execution: (a) either all the newly added pairs corresponding to the assignment $t \assgn s$ complied to the regular structure, or (b) each new pair $[t', s']$ that did not fit the regular pattern was later removed via a subsequent instruction ``$\create{u}$'' or ``$\forget{u}$'' within $p$, with $u$ a prefix of $t'$ or $s'$.

If the case (a) above was satisfied, then, based on the definition of dot-completeness, a pair
\[
(t,s') \in (\tilde{r_1} - t)[t = \tilde{r_1}\slash s - t] - ot,
\]
where 
\[
\tilde{r_1} = \tilde{r}[ot = t]
\]
cannot break the regular pattern of the alias expressions either.
For the case (b) above, all the ``non-well-behaved'' new pairs will be again removed via a subsequent ``$\create{u}$'' or ``$\forget{u}$'' within $p$.

Therefore, the assumption in~(\ref{eq:red-absurdum}) is false, so it holds that:
\[
\overline{r} \exec p = r \exec p^{n+1} \in \reg{r}{r'}.
\]
\end{itemize}
\end{proof}

\section{A {\K}-machinery for collecting aliases}
\label{sec:implem-k}

In this section we provide the specification of a RL-based mechanism collecting the alias information in the {\K} semantic framework~\cite{DBLP:journals/jlp/RosuS10}. We choose {\K} more as a notational convention to enable compact and modular definitions. In reality, the {\K}-rules in this section are implemented in Maude, as rewriting theories, on top of the formalization of SCOOP~\cite{DBLP:conf/acsd/MorandiSNM13} (we refer to Section~\ref{sec:alias-SCOOP} for more details on our approach).

In short, our strategy is to start with a program built on top of the control structures in~(\ref{eq:BNF-control-struct}), then to apply the corresponding {\K}-rules in order to get the ``may aliasing'' information in a designated {\K}-cell ($\cellS{-}{\textnormal{al}}$). Independently of the setting (sequential or concurrent) one can exploit this approach in order to evaluate the aliases of a given finite length $L$. We also show that for sequential contexts, the application of the {\K}-rules is finite and the aliases in the final configuration soundly over-approximate the (infinite) ``may alias'' relations of the calculus.

\begin{paragraph}{Brief overview of {\K}}{
{\K}~\cite{DBLP:journals/jlp/RosuS10} is an executable semantic framework based on Rewriting Logic~\cite{DBLP:conf/fct/MeseguerR11}. It is suitable for defining (concurrent) languages and corresponding formal analysis tools, with straightforward implementation in {\K}-Maude~\cite{DBLP:conf/wrla/SerbanutaR10}. {\K}-definitions make use of the so-called \emph{cells}, which are labelled and can be nested, and (rewriting) \emph{rules} describing the intended (operational) semantics.

A \emph{cell} is denoted by $\cellS{-}{\textnormal{[name]}}$, where [name] stands for the \emph{name of the cell}. A construction $\cellS{.}{\textnormal{n}}$ stands for an \emph{empty cell} named n. We use ``pattern matching'' and write $\cellSL{c}{\textnormal{n}}$ for a cell with content $c$ at the top, followed by an arbitrary content ($\ldots$). Orthogonally, we can utilize cells of shape $\cellSR{c}{\textnormal{n}}$ and $\cellS{\ldots c \ldots}{{\textnormal{n}}}$, defined in the obvious way.

Of particular interest is $\cellS{-}{\textnormal{k}}$ -- the \emph{continuation cell}, or the \emph{$k$-cell}, holding the stack of program instructions (associated to one processor), in the context of a programming language formalization. We write
\[
\cellSL{i_1  \curvearrowright i_2}{\textnormal{k}}
\]
for a set of instructions to be ``executed'', starting with instruction $i_1$, followed by $i_2$. The associative operation $ \curvearrowright$ is the instruction sequencing.

A {\K}-rewrite rule
\begin{equation}
\label{eq:def-k-rew-rule-ex}
\cellSL{c}{\textnormal{$n_1$}}
\cellS{c'}{\textnormal{$n_2$}}
~\Rightarrow~
\cellSL{c'}{\textnormal{$n_1$}}
\cellSR{c'}{\textnormal{$n_3$}}
\end{equation}
reads as: if cell \textnormal{${n_1}$} has $c$ at the top and cell \textnormal{${n_2}$} contains value $c'$, 
then $c$ is replaced by $c'$ in \textnormal{${n_1}$} and $c'$ is added at the end of the cell \textnormal{${n_3}$}. The content of \textnormal{${n_2}$} remains unchanged.
}
In short,~(\ref{eq:def-k-rew-rule-ex}) is written in a {\K}-like syntax as:
\[
\cellUL{c}{c'}{\textnormal{$n_1$}}
\cellS{c'}{\textnormal{$n_2$}}
\cellUR{.}{c'}{\textnormal{$n_3$}}.
\]
\end{paragraph}

We further provide the details behind the {\K}-specification of the alias calculus. As expected, the $k$-cell retains the instruction stack of the object-oriented program.
We utilize cells
$
\langle - \rangle_{\textnormal{al}}
$
to enclose the current alias information, and the so-called \emph{back-tracking cells} 
$
\langle - \rangle_{\textnormal{bkt-\ldots}}
$
enabling the sound computation of aliases for the case of \mbox{$\thenelse{\!\!-\!\!}{\!\!-\!\!}$} and, in non-concurrent contexts, for loops and (possibly recursive) function calls.
As a convention, we mark with ($\clubsuit$) the rules that are sound only for non-concurrent applications, based on Lemma~\ref{lm:reg-expr}.

The following {\K}-rules are straightforward, based on the axioms~(\ref{eq:def-seq-comp})--(\ref{eq:def-assign}) in Section~\ref{sec:ext-inf-expr}. 
Namely, the rule implementing an instruction $p \,;\, q$ simply forces the sequential execution of $p$ and $q$ by positioning $p \curvearrowright q$ at the top of the continuation cell:

\begin{equation}
\label{appeq:k-seq-comp}
\cellUL{p\,;\,q}{p \curvearrowright q }{k}
\end{equation}

Handling $\create x$ and $\forget x$ complies to the associated definitions. Namely, it updates the current alias relation by removing all the pairs having (at least) one element with $x$ as prefix. In addition, it also pops the corresponding instruction from the continuation stack:

\begin{equation}
\label{appeq:k-create-forget}
\begin{array}{lr}
\cellU{r}{r-x}{al}\!\!\cellUL{\create x}{.}{k} \hspace{15pt}
&
\cellU{r}{r-x}{al}\!\!\cellUL{\forget x}{.}{k}
\end{array}
\end{equation}

The assignment rule restores the current alias relation according to its axiom in~(\ref{eq:def-assign}), and removes the assignment instruction from the top of the $k$-cell:

\begin{equation}
\label{appeq:k-assign}
\cellU{r}{(r_{1} - t)[t\,=\,(r_1 \slash s \,-\,t)] - ot}{al}\cellUL{t \assgn s}{.}{k}\,\,\,\,\,
\textnormal{with~} r_{1} = r[ot = t]
\end{equation}

The {\K}-implementation of a $\thenelse{p}{q}$ statement is more sophisticated, as it instruments a stack-based mechanism enabling the computation of the union of alias relations $r \exec p \,\cup\, r\exec q$ in three steps. First, we define the {\K}-rule:
\begin{equation}
\label{appeq:k-then-else-first}
\cellS{r}{al}
\cellUL{\thenelse{p}{q}}{p \mfbox{et} q \mfbox{ee}}{k}
\cellUL{.}{\cellS{r,p}{t}\,\, \cellS{r,q}{e}}{bkt-te}
\end{equation}
saving at the top of the back-tracking stack $\langle - \rangle_{\textnormal{bkt-te}}$ the initial alias relation $r$ to be modified by both $p$ and $q$, via two cells $\langle r, p\rangle_{\textnormal{t}}$ and $\langle r, q\rangle_{\textnormal{e}}$, respectively. Note that the original instruction in the $k$-cell is replaced by a meta-construction marking the end of the executions corresponding to the {\bf then} and {\bf else} branches with $\mfbox{et}$ and $\mfbox{ee}$, respectively. 

Second, whenever the successful execution of $p$ (signaled by $\mfbox{et}$) at the top of the $k$-cell) builds an alias relation $r'$, the execution of $q$ starting with the original relation $r$ is forced by replacing $r'$ with $r$ in $\langle - \rangle_{\textnormal{al}}$, and by positioning $q \mfbox{ee}$ at the top of the $k$-cell. The new alias information after $p$, denoted by $\langle r', p\rangle_{\textnormal{t}}$, is updated in the back-tracking cell: 

\begin{equation}
\label{appeq:k-end-then}
\cellU{r'}{r}{al}
\cellUL{\mfbox{et} q \mfbox{ee}}{q \mfbox{ee}}{k}
{
\begin{tabular}{r@{}c@{}c@{}l}
\small
$\langle\,\,$ & $\cellS{r,p}{t}$ & $\cellS{r,q}{e}$ & $ \,\,\ldots\rangle_{\textnormal{bkt-te}}$\\
\cline{2-2}
& $\cellS{r',p}{t}$ & &\\
\end{tabular}
}
\end{equation}

Eventually, if the successful execution of $q$ (marked by $\mfbox{ee}$ at the top of $\langle-\rangle_{\textnormal k}$) produces an alias relation $r''$, then the final alias information becomes $r' \,\cup\, r''$, where $r'$ is the aliasing after $p$, stored as showed in~(\ref{appeq:k-end-then}). The corresponding back-tracking information is removed from  $\langle - \rangle_{\textnormal{bkt-te}}$, and the next program instruction is enabled in the $k$-cell:

\begin{equation}
\label{appeq:k-end-else}
\cellU{r''}{r'\,\cup\, r''}{al}
\cellUL{\mfbox{ee}}{.}{k}
\cellUL{\cellS{r',p}{t}\,\,\cellS{r,q}{e}}{.}{bkt-te}
\end{equation}

For $\loopend{p}$, we utilize a meta-construction $p \mfbox{l} \loopend{p}$ simulating the set union in~(\ref{eq:def-loop}), and a back-tracking stack $\langle - \rangle_{\textnormal{bkt-l}}$ collecting the alias information obtained after each execution of $p$. Moreover, the {\K}-implementation exploits the result in Lemma~\ref{lm:reg-expr}.
Whenever a ``lasso'' is reached, the infinite rewriting is prevented by resuming the infinite application of $p$ in terms of a sound over-approximating alias relation. The {\K}-rules are as follows.

First, the aforementioned unfolding is performed, and the alias relation before $p$ is stored in the back-tracking cell as $\langle r \rangle_{\textnormal{al-o}} \langle p \rangle_{\textnormal{l}}$:

\begin{equation}
\label{appeq:k-loop-whole}
\cellS{r}{al}
\cellUL{\loopend{p}}{p \mfbox{l} \loopend{p}}{k}
\cellUL{.}{\cellS{r}{al-o}\cellS{p}{l}}{bkt-l}
\end{equation}

If the alias relation $r'$ obtained after the successful execution of $p$ (marked by $\mfbox{l}$ at the top of the continuation) is not a lasso of the aliasing $r$ before $p$ (previously stored in $\langle -\rangle_{\textnormal{bkt-l}}$) then $p$ is constrained to a new execution by becoming the top of the $k$-cell, and $r'$ is memorized for back-tracking:

\begin{equation}
\label{appeq:k-loop-not-lasso}
\cellS{r'}{al}
\cellUL{\mfbox{l} \loopend{p}}{p \mfbox{l} \loopend{p}}{k}
\cellUL{\cellS{r}{al-o}\cellS{p}{l}}{\cellS{r'}{al-o}\cellS{p}{l}}{bkt-l}
\textnormal{~if not~} \lasso{r}{r'}~~(\clubsuit)
\end{equation}

Last, if a lasso is reached after the execution of $p$, then the current aliasing is soundly replaced by a ``regular'' over-approximation $\reg{r}{r'}$, the corresponding back-tracking information is removed from $\langle -\rangle_{\textnormal{bkt-l}}$ and the {\bf loop} instruction is eliminated from the $k$-cell:

\begin{equation}
\label{appeq:k-loop-lasso}
\cellU{r'}{\reg{r}{r'}}{al}\!\!
\cellUL{\mfbox{l} \loopend{p}}{.}{k}\!\!
\cellUL{\cellS{r}{al-o}\cellS{p}{l}}{.}{bkt-l}
\textnormal{~if~} \lasso{r}{r'}~~(\clubsuit)
\end{equation}

For handling function calls such as $\call{f(l)}$ we use a meta-construction $\mid f \mid \!\! \mfbox{f}$. Here $\mid f \mid$ stands for the body of $f$ and $\mfbox{f}$ marks the end of the corresponding execution. Moreover, a stack $\langle - \rangle_{\textnormal{bkt-cf}}$ is utilized in order to store the alias information before each (possibly recursive) call of $f$, with the purpose of identifying the  lassos generated by the (possibly repeated) execution of $f$.
In order to guarantee a sound implementation of (mutually) recursive calls,
both $\mfbox{f}$ and  $\langle - \rangle_{\textnormal{bkt-cf}}$ are parameterized by $f$ -- the name of the function.
An example illustrating this reasoning mechanism is provided in Section~\ref{sec:example-k-machinery}.

The first {\K} rule for handling function calls matches the associated axiom in~(\ref{eq:def-qualified-call}): the alias information is set to $r[f^{\bullet}:l]$, whereas the next instructions to be executed are given by $\mid f \mid$. Note that the original aliasing is retained in the (initially empty) back-tracking cell via $\langle r \rangle_{\textnormal{al-o}}$.

\begin{equation}
\label{appeq:k-call-first}
\cellU{r}{r[f^\bullet:l]}{al}
\cellUL{\call{f(l)}}{\mid f \mid \mfbox{f}}{k}
\cellU{.}{\cellS{r}{al-o}}{bkt-cf}
\end{equation}

\begin{remark}
\label{apprem:formal-vs-actual}
Observe that the back-tracking cell does not need to be parameterized by the actual argument list $l$ of $f$. Each such argument is anyways replaced in the current alias relation $r$ by its counterpart in the formal argument list of $f$. In short: $r$ becomes $r[f^\bullet : l]$.
\end{remark}

A successful execution of $\call{f(l)}$ is distinguished by the occurrence of $\mfbox{f}$ at the top of the continuation stack. If this is the case, then the corresponding back-tracking alias information is removed from $\langle - \rangle_{\textnormal{bkt-cf}}$ and the next program instruction (if any) is enabled at the top of the $k$-cell:

\begin{equation}
\label{appeq:k-call-exit}
\cellS{r'}{al}
\cellUL{\mfbox{f}}{.}{k}
\cellUL{\cellS{r}{al-o}}{.}{bkt-cf}
\end{equation}

Recursive calls are treated by means of two {\K}-rules. Note that a recursive context is identified whenever the current program instruction is of shape $\call{f(l)}$ and the associated back-tracking structure is not empty, {\it i.e.}, rule~(\ref{appeq:k-call-first}) was previously applied.
Then, if the recursive call of $f$ when starting with $r$ produces a lasso $r'$, the execution of $f(l)$ is stopped by soundly over-approximating the alias information with $\reg{r}{r'}$, according to Lemma~\ref{lm:reg-expr}, and by removing $\call{f(l)}$ from the $k$-cell:

\begin{equation}
\label{appeq:k-call-lasso}
\cellU{r'}{\reg{r}{r'}}{al}
\cellUL{\call{f(l)}}{.}{k}
\cellSL{\cellS{r}{al-o}}{bkt-cf}
\textnormal{~if~} \lasso{r}{r'}~~(\clubsuit)
\end{equation}

If a lasso is not reached, then the body of $f$ is executed once more, and the current aliasing is pushed to the back-tracking cell:

\begin{equation}
\label{appeq:k-call-not-lasso}
\cellS{r'}{al}
\cellUL{\call{f(l)}}{\mid f \mid \mfbox{f}}{k}
{
\begin{tabular}{r@{}c@{}c@{}l}
\small
$\langle\,\,$ & $.$ & $\cellS{r}{al-o}$ & $ \,\,\ldots\rangle_{\textnormal{bkt-cf}}$\\
\cline{2-2}
& $\cellS{r'}{al-o}$ & &\\
\end{tabular}
}
\textnormal{~if~not~} \lasso{r}{r'}~~(\clubsuit)
\end{equation}

Qualified calls $x.\call{f(l)}$ are handled by two {\K}-rules as follows. First, based on the definition in~(\ref{eq:def-qualified-call}), the ``negative variable'' $x'$ transposing the context of the call to to the context of the caller is distributed to the elements of the initial alias relation $r$, and to $l$ -- the argument list of $f$. Moreover, a meta-construction $\mfbox{qf}$ is utilized in order to mark the end of the qualified call in the continuation cell, similarly to the rule~(\ref{appeq:k-call-first}). The caller is stored in a back-tracking stack $\cellS{.}{\textnormal{bkt-qf}}$ also parameterized by $f$ -- the name of the function. The current instruction in the $k$-cell becomes $\call{f(x'.l)}$, as expected:

\begin{equation}
\label{appeq:k-qcall-first}
\cellU{r}{x'.r}{al}
\cellUL{x. \call{f(l)}}{\call{f(x'.l)} \mfbox{qf}}{k}
\cellU{.}{\cellS{x}{f}}{bkt-qf}
\end{equation}

Second, when the successful termination of the qualified call is signaled by $\mfbox{qf}$ at the top of the $k$-cell, the corresponding stored caller is distributed to the current alias relation and removed from the back-tracking cell. The next instruction in the continuation cell is released by eliminating the top $\mfbox{qf}$: 

\begin{equation}
\label{appeq:k-qcall-exit}
\cellU{r}{x.r}{al}
\cellUL{\mfbox{qf}}{.}{k}
\cellUL{\cellS{x}{f}}{.}{bkt-qf}
\end{equation}

In a non-concurrent setting, the machinery orchestrating the {\K}-rules introduced in this section implements an algorithm that always terminates and provides a sound over-approximation of ``may aliasing''.

\begin{theorem}
\label{th:dec-proc}
Consider $p$ a program built on top of the control structures in~(\ref{eq:BNF-control-struct}), that executes in a sequential setting. Then, the application of the corresponding {\K}-rules when starting with $p$ and an empty alias relation, is a finite rewriting of shape
\[
\cellS{\emptyset}{\textnormal{al}}\cellS{p}{\textnormal{k}}\,\,\xRightarrow{(*)}\,\,\cellS{r}{\textnormal{al}}\cellS{.}{\textnormal{k}},
\]
with $r$ a sound over-approximation of the aliasing information corresponding to the execution of $p$.
\end{theorem}
\begin{proof}
The key observation is that, due to the execution of loops and/or recursive calls, expressions can infinitely grow in a \emph{regular} fashion. Hence, a lasso is always reached. Consequently, the control structure generating the infinite behaviour is removed from the $k$-cell, according to the associated {\K}-specification for loops and/or recursive calls. This guarantees termination. Moreover, recall that the regular expressions replacing the current alias information are a sound over-approximation, according to Lemma~\ref{lm:reg-expr}.
\end{proof}

Observe that the $RL$-based machinery can simulate precisely the ``cutting at length $L$'' technique in~\cite{Meyer-aliasing-13}. It suffices to disable the rules ($\clubsuit$) and stop the rewriting after L steps.

The naturalness of applying the resulted aliasing framework is illustrated in the example in Section~\ref{sec:example-k-machinery}, for the case of two mutually recursive functions.

\subsection{The {\K}-machinery by example}
\label{sec:example-k-machinery}

For an example, in this section we show how the {\K}-machinery developed in Section~\ref{sec:implem-k} can be used in order to extract the alias information for the case of two mutually recursive functions defined as:
\[
\begin{array}{lr}
\begin{array}{l}
f(x) \,\, \{
\,\, x \assgn x.a\,; 
~~\call{g(x)} \,\, \}
\end{array}
&\hspace{40pt}
\begin{array}{l}
g(x) \,\, \{
\,\, x \assgn x.b\,;
~~\call{f(x) \,\, \}}
\end{array}
\end{array}
\]
We assume that $x$ is an object of a class with two fields $a$ and $b$, respectively. We consider a sequential setting.

At first glance it is easy to see that the execution of $\call{f(x)}$, when starting with an empty alias relation $r$, produces the alias expressions:
\begin{equation}
\label{eq:ex-mut-rec-alias}
[x,\,x.(a.b)^*]~~~[x.a,\,x.(a.b)^*.a]~~~[x.b,\,x.(a.b)^*.b]
\end{equation}

The associated reasoning in {\K} is depicted in Figure~\ref{K-machinery-example}. The whole procedure starts with an empty alias relation $r = \emptyset$, and $\call{f(x)}$ in the continuation stack. Then, the corresponding {\K} rules (for handling assignments and function calls) are applied in the natural way.

A lasso is reached after two calls of $f(x)$ that, consequently, determine two calls of $g(x)$ -- identified by $\mfbox{g} \mfbox{f} \mfbox{g} \mfbox{f}$ in the $k$-cell. This triggers the application of rule~(\ref{appeq:k-call-lasso})  enabling the ``regular'' over-approximation as in Lemma~\ref{lm:reg-expr}.

Our example also illustrates the importance of isolating the back-traced alias information in cells of shape $\cellS{.}{\textnormal{bkt-cf}}$ parameterized by the (possibly recursive) function $f$. More explicitly, rule~(\ref{appeq:k-call-lasso}) is soundly applied by identifying the aforementioned lasso based on: the current alias relation $r_4$, the recursive call $f(l)$ at the top of the continuation, and the back-traced aliasing $\cellSL{\cellS{r_2}{\textnormal{al-o}}}{\textnormal{bkt-cf}}$ associated to the previous executions of $f(l)$.

\renewcommand{\arraystretch}{1.5}
\setlength{\tabcolsep}{2pt}

As introduced in~(\ref{eq:def-lasso}), an alias relation $r'$ is a lasso of a relation $r$ whenever there is a one-to-one correspondence between their elements as follows:
\[
[x_1 y_1 z_1, x_2 y_2 z_2] \in r \textnormal{~~iff~~} [x_1 y_1 y_1 z_1, x_2 y_2 y_2 z_2] \in r'.
\]
The current alias relation
\[
r_4 = \{[x, x.a.b.a.b],\, [x.a, x.a.b.a.b.a],\, [x.b, x.a.b.a.b.b]\},
\]
before applying rule~(\ref{appeq:k-call-lasso}), is a lasso of
\[
r_2 = \{[x, x.a.b],\, [x.a, x.a.b.a],\, [x.b, x.a.b.b]\}.
\]
The aforementioned one-to-one correspondence is summarized in the following table:
\[
\footnotesize
\begin{tabular}{r@{}c@{}l|c|c|c|c|c|c}
$[x_1 y_1 z_1, x_2 y_2 z_2] \in r_2$& {~iff~} & $ [x_1 y_1 y_1 z_1, x_2 y_2 y_2 z_2] \in r_4$ & $x_1$ & $y_1$ & $z_1$ & $x_2$ & $y_2$ & $z_2$\\
\hline
$[x, x.a.b] \in r_2$ & {~iff~} & $[x, x.a.b.a.b] \in r_4$ & $x$ & $\varepsilon$ & $\varepsilon$ & $x$ & $a.b$ & $\varepsilon$\\
\hline
$[x.a, x.a.b.a] \in r_2$ & {~iff~} & $[x.a, x.a.b.a.b.a] \in r_4$ & $x$ & $\varepsilon$ & $a$ & $x$ & $a.b$ & $a$\\
\hline
$[x.b, x.a.b.b] \in r_2$ & {~iff~} & $[x.b, x.a.b.a.b.b] \in r_4$ & $x$ & $\varepsilon$ & $b$ & $x$ & $a.b$ & $b$
\end{tabular}
\]
Here $\varepsilon$ stands for the \emph{empty alias expression}.

Moreover, according to rule~(\ref{appeq:k-call-lasso}), the lasso shaped by $r_2$ and $r_4$ also causes the (otherwise infinite) recursive calls to stop, as $\call{f(l)}$ is eliminated from the top of the $k$-cell.
Hence, the rewriting process finishes with
a sound over-approximation $\reg{r_2}{r_4}$ replacing the current alias relation (cf. Lemma~\ref{lm:reg-expr}), defined precisely as in~(\ref{eq:ex-mut-rec-alias}). 

\renewcommand{\arraystretch}{0.5}

\begin{figure}[!htbp]
\label{K-machinery-example}
\[
\begin{tabular}{rl}
\label{fig:ex-mut-rec-k}
$\cellS{r}{\textnormal{al}}$ & $\cellS{\call{f(x)}}{\textnormal{k}}$\\
$\cellS{.}{\textnormal{bkt-cf}}$ & $\cellS{.}{\textnormal{bkt-cg}}$
\\[1.5ex]
\multicolumn{2}{c}{$\Downarrow$~(\ref{appeq:k-call-first})}\\[1.5ex]
$\cellS{r}{\textnormal{al}}$ & $\cellS{x \assgn x.a;\, \call{g(x)} \mfbox{f}}{\textnormal{k}}$\\
$\cellS{\cellS{r}{\textnormal{al-o}}}{\textnormal{bkt-cf}}$ & $\cellS{.}{\textnormal{bkt-cg}}$
\\[1.5ex]
\multicolumn{2}{c}{$\Downarrow$~(\ref{appeq:k-assign})}\\[1.5ex]
$\cellS{r_1}{\textnormal{al}}$ & $\cellS{\call{g(x)} \mfbox{f}}{\textnormal{k}}$\\
$\cellS{\cellS{r}{\textnormal{al-o}}}{\textnormal{bkt-cf}}$ & $\cellS{.}{\textnormal{bkt-cg}}$\\
\multicolumn{2}{c}{where $r_1 = \{[x, x.a],\, [x.a, x.a.a],\, [x.b, x.a.b]\}$}
\\[1.5ex]
\multicolumn{2}{c}{$\Downarrow$~(\ref{appeq:k-call-not-lasso})}\\[1.5ex]
$\cellS{r_1}{\textnormal{al}}$ & $\cellS{x \assgn x.b;\, \call{f(x)} \mfbox{g} \mfbox{f}}{\textnormal{k}}$\\
$\cellS{\cellS{r}{\textnormal{al-o}}}{\textnormal{bkt-cf}}$ & $\cellS{\cellS{r_1}{\textnormal{al-o}}}{\textnormal{bkt-cg}}$
\\[1.5ex]
\multicolumn{2}{c}{$\Downarrow$~(\ref{appeq:k-assign})}\\[1.5ex]
$\cellS{r_2}{\textnormal{al}}$ & $\cellS{\call{f(x)} \mfbox{g} \mfbox{f}}{\textnormal{k}}$\\
$\cellS{\cellS{r}{\textnormal{al-o}}}{\textnormal{bkt-cf}}$ & $\cellS{\cellS{r_1}{\textnormal{al-o}}}{\textnormal{bkt-cg}}$\\
\multicolumn{2}{c}{where $r_2 = \{[x, x.a.b],\, [x.a, x.a.b.a],\, [x.b, x.a.b.b]\}$}
\\[1.5ex]
\multicolumn{2}{c}{$\Downarrow$~(\ref{appeq:k-call-not-lasso})}\\[1.5ex]
$\cellS{r_2}{\textnormal{al}}$ & $\cellS{ x \assgn x.a;\,\call{g(x)} \mfbox{f} \mfbox{g} \mfbox{f}}{\textnormal{k}}$\\
$\cellS{\cellS{r_2}{\textnormal{al-o}}\,\, \cellS{r}{\textnormal{al-o}}}{\textnormal{bkt-cf}}$ & $\cellS{\cellS{r_1}{\textnormal{al-o}}}{\textnormal{bkt-cg}}$
\\[1.5ex]
\multicolumn{2}{c}{$\Downarrow$~(\ref{appeq:k-assign})}\\[1.5ex]
$\cellS{r_3}{\textnormal{al}}$ & $\cellS{ \call{g(x)} \mfbox{f} \mfbox{g} \mfbox{f}}{\textnormal{k}}$\\
$\cellS{\cellS{r_2}{\textnormal{al-o}}\,\, \cellS{r}{\textnormal{al-o}}}{\textnormal{bkt-cf}}$ & $\cellS{\cellS{r_1}{\textnormal{al-o}}}{\textnormal{bkt-cg}}$\\
\multicolumn{2}{c}{where $r_3 = \{[x, x.a.b.a],\, [x.a, x.a.b.a.a],\, [x.b, x.a.b.a.b]\}$}
\\[1.5ex]
\multicolumn{2}{c}{$\Downarrow$~(\ref{appeq:k-call-not-lasso})}\\[1.5ex]
$\cellS{r_3}{\textnormal{al}}$ & $\cellS{ x\assgn x.b;\,\call{f(x)} \mfbox{g} \mfbox{f} \mfbox{g} \mfbox{f}}{\textnormal{k}}$\\
$\cellS{\cellS{r_2}{\textnormal{al-o}}\,\, \cellS{r}{\textnormal{al-o}}}{\textnormal{bkt-cf}}$ & $\cellS{\cellS{r_3}{\textnormal{al-o}}\,\ \cellS{r_1}{\textnormal{al-o}}}{\textnormal{bkt-cg}}$
\\[1.5ex]
\multicolumn{2}{c}{$\Downarrow$~(\ref{appeq:k-assign})}\\[1.5ex]
$\cellS{r_4}{\textnormal{al}}$ & $\cellS{ \call{f(x)} \mfbox{g} \mfbox{f} \mfbox{g} \mfbox{f}}{\textnormal{k}}$\\
$\cellS{\cellS{r_2}{\textnormal{al-o}}\,\, \cellS{r}{\textnormal{al-o}}}{\textnormal{bkt-cf}}$ & $\cellS{\cellS{r_3}{\textnormal{al-o}}\,\ \cellS{r_1}{\textnormal{al-o}}}{\textnormal{bkt-cg}}$\\
\multicolumn{2}{c}{where $r_4 = \{[x, x.a.b.a.b],\, [x.a, x.a.b.a.b.a],\, [x.b, x.a.b.a.b.b]\}$}
\\[1.5ex]
\multicolumn{2}{c}{$\Downarrow$~(\ref{appeq:k-call-lasso})}\\[1.5ex]
$\cellS{reg(r_2, r_4)}{\textnormal{al}}$ & $\cellS{ \mfbox{g} \mfbox{f} \mfbox{g} \mfbox{f}}{\textnormal{k}}$\\
$\cellS{\cellS{r_2}{\textnormal{al-o}}\,\, \cellS{r}{\textnormal{al-o}}}{\textnormal{bkt-cf}}$ & $\cellS{\cellS{r_3}{\textnormal{al-o}}\,\ \cellS{r_1}{\textnormal{al-o}}}{\textnormal{bkt-cg}}$\\
\\[1.5ex]
\multicolumn{2}{c}{$\Downarrow$~(*)(\ref{appeq:k-call-exit})}\\[1.5ex]
\multicolumn{2}{c}{
$\cellS{\{[x, x.(a.b)^*],\,[x.a, x.(a.b)^*.a],\, [x.b, x.(a.b)^*.b]  \}}{\textnormal{al}} \cellS{.}{\textnormal{k}}
\cellS{.}{\textnormal{bkt-cf}} \cellS{.}{\textnormal{bkt-cg}}$}
\end{tabular}
\]
\caption{Aliasing and mutual recursion in {\K}.}
\end{figure}

\section{Aliasing in SCOOP}
\label{sec:alias-SCOOP}

In this section we provide a brief overview on the integration and applicability of the alias calculus in SCOOP.
First, recall from Section~\ref{sec:SCOOP} that the Maude semantics of SCOOP in~\cite{DBLP:conf/acsd/MorandiSNM13} is defined over tuples of shape   
\[
\langle p_1 \,::\, St_{1} \mid \ldots \mid p_n \,::\, St_{n}, \sigma \rangle
\]
where, $p_i$ and $St_i$ stand for processors and their call stacks, respectively. $\sigma$ is the state of the system and it holds information about the {heap} and the {store}.

The assignment instruction, for instance, is formally specified as the transition rule:
\begin{equation}
\label{eq:assign-Morandi}
\dfrac{\textnormal{a is fresh}}{
\Gamma \vdash \langle p \,::\, t\,:=s;\, St,\, \sigma \rangle \rightarrow
\langle p\,::\, \textnormal{eval}(a, s);\, \textnormal{wait}(a);\, \textnormal{write}(t, a.data);\,St,\, \sigma\rangle
}
\end{equation}
where, intuitively, ``eval$(a, s)$'' evaluates $s$ and puts the result on channel $a$, ``wait$(a)$'' enables processor $p$ to use the evaluation result, ``write$(t, a.data)$'' sets the value of $t$ to $a.data$, $St$ is a call stack, and $\Gamma$ is a typing environment~\cite{nienaltowski2007practical} containing the class hierarchy of a program and all the type definitions.

At this point it is easy to understand that the {\K}-rule for assignments
\[
\cellU{r}{(r_{1} - t)[t\,=\,(r_1 \slash s \,-\,t)] - ot}{al}\cellUL{t \assgn s}{.}{k}\,\,\,\,\,
\textnormal{with~} r_{1} = r[ot = t]~~~~(\ref{appeq:k-assign})
\]
can be straightforwardly integrated in~(\ref{eq:assign-Morandi}) by enriching the SCOOP configuration with a new component $alias\_$ encapsulating the alias information, and considering instead the transition:
\[
\begin{array}{cc}
\Gamma \vdash \langle p \,::\,\, t\, {:=} s;\, St,\, \sigma, alias_{old} \rangle \rightarrow\\
\langle p\,::\, \textnormal{eval}(a, s);\, \textnormal{wait}(a);\, \textnormal{write}(t, a.data);\,St,\, \sigma, alias_{new}\rangle
\end{array}
\]
where 
\[
\begin{array}{lr}
alias_{old} = r ~~~~~~~~~
alias_{new} =(r_{1} - t)[t\,=\,(r_1 \slash s \,-\,t)] - ot 
\end{array}
\]
with $r$ and $r_1$ as in~(\ref{appeq:k-assign}).
The integration of all the {\K}-rules of the alias calculus on top of the Maude formalization of SCOOP is achieved by following a similar approach.

For a case study, one can download the SCOOP formalization at:\\
{\verb+https://dl.dropboxusercontent.com/u/1356725/SCOOP-SCP.zip+}\\
and run the command\\
{ \verb+> maude SCOOP.maude ..\examples\linked_list-test.maude+}\\
corresponding to the execution of the code in~(\ref{eq:intor-ex})
\begin{equation}
\notag
\begin{array}{l}
\create{x_0}\\
\loopend{
\\\hspace{10pt}i\,:=i+1
\\\hspace{10pt}\create{x_i}
\\\hspace{10pt}x_i . set\_next(x_{i-1})\\\hspace{-4.5pt}
}
\end{array}
\end{equation}
for five iterations of the loop.
As can be observed based on the code in {\verb+aliasing-linked_list.maude+}, in order to implement our applications in Maude, we use intermediate (still intuitive) representations.
For instance, the class structure defining a node in a simple linked list, with filed \emph{next} and setter $set\_next$ is declared as:
{
\begin{verbatim}
class 'NODE
    create {'make} (
    attribute { 'ANY } 'next : [?, . , 'NODE] ;
    procedure { 'ANY } 'set_next ( 'a_next : [?, ., 'NODE] ;) [...]
   )
end ;
\end{verbatim}
}
\noindent
where {\verb+'next : [?, . , 'NODE]+} stands for an object of type {\verb+NODE+}, that is handled by the current processor ({\verb+.+}) and that can be Void ({{\verb+?+}), and {\verb+'make+} plays the role of a constructor.
The ``entry point'' of the program corresponds to the function {\verb+'make+} in the (main) class {\verb+'LINKED_LIST_TEST+} and is set via:
{
\begin{verbatim}
settings('LINKED_LIST_TEST, 'make, aliasing-on) .
\end{verbatim}
}
\noindent
Observe that the flag for performing the alias analysis is switched to ``on''.

The relevant parts of the corresponding Maude output after executing the aforementioned command are as follows:
{
\begin{verbatim}
                     \||||||||||||||||||/
                   --- Welcome to Maude ---
                     /||||||||||||||||||\
            Maude 2.6 built: Mar 31 2011 23:36:02
   
rewrite in SYSTEM :
[...] settings('LINKED_LIST_TEST, 'make, aliasing-on))

|-
  {0}proc(0) :: nil | {0}proc(1) :: nil,
  {['x0 ; 'x0]} U {['x0 ; 'x1.'next]} U
  {['x0 ; 'x2.'next.'next]} U {['x0 ; 'x3.'next.'next.'next]} U
  {['x0 ; 'x4.'next.'next.'next.'next]} U
  [...]
  {['x3 ; 'x3]} U {['x3 ; 'x4.'next]}

state
  heap [...]
  store [...]
end
\end{verbatim}
}

In short, one can see that two processors that were created finished executing the instructions on their corresponding call stacks:
{\verb+{0}proc(0) :: nil+} and {\verb+{0}proc(1) :: nil+}. The aliased expressions include, as expected based on~(\ref{eq:aliasing-example}), pairs of shape $[x_i\, ;\, x_{i+k}.next^k]$. Moreover, the output displays the contents of the current system state, by providing information on the \emph{heap} and \emph{store}, as formalized in~\cite{DBLP:conf/acsd/MorandiSNM13}.

\section{Deadlocking in SCOOP}
\label{sec:coffman-deadlocks}

In what follows we provide the details behind the formalization and the implementation of a deadlock detection mechanism for SCOOP. We discuss how Maude can be exploited in order to test and, respectively, model-check deadlocks in the context of SCOOP programs, we analyze the associated challenges and propose a series of corresponding solutions.

\subsection{Formalizing deadlocks in SCOOP}
\label{sec:formalizing-deadlocks}

Recall that the key idea of SCOOP is to associate to each object declared as \emph{separate} a processor, or handler in charge of executing the routines of that object. Assume a processor $p$ that performs a call $o.f(a_1, a_2, \ldots)$ on a separate object $o$, with separate arguments $a_i$ ($i \geq 1$). As previously stated in Section~\ref{sec:SCOOP}, in the SCOOP semantics, the application of the call  $f(\ldots)$ will \emph{wait} until it has been able to \emph{lock} all the separate objects associated to $a_1, a_2, \ldots$. This reservation mechanism enables deadlocks to occur whenever a set of processors reserve each other circularly. This situation might happen, for instance, in a Dining Philosophers scenario, where both philosophers and forks are objects residing on their own processors.

\begin{definition}[Deadlock]
\label{def:deadlock}
For a processor $p$, let $W(p)$ be the set of handlers $p$ \emph{waits} for its asynchronous execution, and $H(p)$ represent the set of resources $p$ already acquired.
A \emph{deadlock} exists if for some set $D$ of processors the following holds:
\begin{equation}
\label{eq:deadlock-def-Meyer}
(\forall p \in D).(\exists p' \in D).(p \not = p') \land (W(p) \cap H(p') = \emptyset).
\end{equation}
\end{definition}

The integration of a deadlock detection mechanism based on Definition~\ref{def:deadlock} on top of the SCOOP formalization in~\cite{DBLP:conf/acsd/MorandiSNM13} is immediate. As already presented in Section~\ref{sec:SCOOP}, the operational semantics of SCOOP is defined over tuples of shape:   
\[
\langle p_1 \,::\, St_{1} \mid \ldots \mid p_n \,::\, St_{n}, \sigma \rangle
\]
where, $p_i$ and $St_i$ stand for processors and their call stacks, respectively, and $\sigma$ is the state of the system.
Given a processor $p'$ as in~(\ref{eq:deadlock-def-Meyer}), the set $H(p')$ corresponds, based on~\cite{DBLP:conf/acsd/MorandiSNM13}, to $\sigma.rq\_locks(p')$.
Whenever the top of the instruction stack of a processor $p$ is of shape $lock(\{q_i,\ldots,q_n\})$, we say that the wait set $W(p)$ is the set of processors $\{q_1, \ldots, q_n\}$.
Hence, assuming a predefined system configuration $\langle deadlock \rangle$, the SCOOP transition rule in Maude corresponding to~(\ref{eq:deadlock-def-Meyer}) can be written as:
\begin{equation}
\label{eq:Maude-deadlock}
\dfrac{
\begin{array}{c}
(\exists D \subseteq \sigma.procs).(\forall p \in D). (\exists p' \in D).(p\not = p') \land\\
(aqs\,:=\ldots \mid p\,::\,lock(\{q_i,\ldots\});St \mid \ldots)\,\,\land\,\,
(\sigma.rq\_locks(p').has(q_i))
\end{array}
}
{
\langle aqs, \sigma\rangle \rightarrow \langle deadlock \rangle
} 
\end{equation}
It is intuitive to guess that $\sigma.procs$ in~(\ref{eq:Maude-deadlock}) returns the set of processors in the system, whereas $aqs$ stands for the list of these processors and their associated instruction stacks (separated by the associative and commutative operator ``$\mid$''~). We use ``$\ldots$'' to represent an arbitrary sequence of processors and processor stacks.

\subsection{Testing deadlocks}
\label{sec:deadlocks-Maude-test}

We implemented~(\ref{eq:Maude-deadlock}) and tested the deadlock detection mechanism on top of the formalization in~\cite{DBLP:conf/acsd/MorandiSNM13} for the Dining Philosophers problem.
A case study considering two philosophers sharing two forks can be run by downloading the SCOOP formalization at:\\
{\verb+https://dl.dropboxusercontent.com/u/1356725/SCOOP-SCP-deadlock.zip+}\\
and executing the command\\
{\verb+> maude SCOOP.maude ..\examples\dining-philosophers-rewrite.maude+}\\
\noindent
The class implementing the \emph{philosopher} concept is briefly given below:
{
{
\begin{verbatim}
(class 'PHILOSOPHER 
    create { 'make } (
        attribute {'ANY} 'left : [!,T,'FORK] ;
        attribute {'ANY} 'right : [!,T,'FORK] ;
        
        procedure { 'ANY } 'make ( 'fl : [!,T,'FORK] ; 
                                   'fr : [!,T,'FORK] ; ) 
            do  ( assign ('left, 'fl) ; assign ('right, 'fr) ; )
            [...]
        end ;

        procedure { 'ANY } 'eat_wrong (nil) 
            do ( command ('Current . 'pick_in_turn('left ;)) ; )
            [...]
        end ;

        procedure { 'ANY } 'pick_in_turn ('f : [!,T,'FORK] ; ) 
            do ( command ('Current . 'pick_two('f ; 'right ;)) ; )
            [...]
        end ;

        procedure { 'ANY } 'pick_two ('fa : [!,T,'FORK] ;
                                      'fb : [!,T,'FORK] ; ) 
            do
                (
                command ('fa . 'use(nil)) ;
                command ('fb . 'use(nil)) ;
                )
           [...]
        end ;
[...] end)
\end{verbatim}
}
}
\noindent
It declares two forks -- {\verb+'left+} and {\verb+'right+} of type {\verb+'FORK+}, that can be handled by any processor ({\verb+T+}) and that cannot be {\verb+Void+} ({\verb+!+}).
Assume two philosophers {\verb+p1+} and {\verb+p2+} (of \emph{separate} type {\verb+PHILOSOPHER+}) and two forks {\verb+f1+} and {\verb+f2+} (of \emph{separate} type {\verb+FORK+}). Moreover, assume that {\verb+'left+} and {\verb+'right+} for  {\verb+p1+} correspond to  {\verb+'f1+} and {\verb+'f2+}. For the case of {\verb+p2+} they correspond to  {\verb+'f2+} and {\verb+'f1+}, respectively.
Asynchronously,  {\verb+p1+} and {\verb+p2+} can execute {\verb+eat_wrong+}, which calls {\verb+pick_in_turn(left)+}. In the context of {\verb+p1+}, the actual value of {\verb+left+} is {\verb+f1+}, whereas for {\verb+p2+} is {\verb+f2+}. Consequently, both resources {\verb+f1+} and {\verb+f2+}, respectively, may be locked ``at the same time'' by {\verb+p1+} and {\verb+p2+}, respectively. Note that {\verb+pick_in_turn+} subsequently calls {\verb+pick_two+} that, intuitively, should enable the philosophers to use both forks. Thus, if {\verb+f1+} and {\verb+f2+}, respectively, are locked by {\verb+p1+} and {\verb+p2+}, respectively, the calls {\verb+pick_two(f2, f1)+} and {\verb+pick_two(f1, f2)+} corresponding to {\verb+p1+} and {\verb+p2+} will (circularly) wait for each other to finish. According to the SCOOP semantics,
{\verb+pick_two(f1, f2)+} is waiting for {\verb+p2+} to release {\verb+f2+}, whereas
{\verb+pick_two(f2, f1)+} is waiting for {\verb+p1+} to release {\verb+f1+}, as the forks are passed to {\verb+pick_two(...)+} as \emph{separate} types. In the context of SCOOP, this corresponds to a Coffman deadlock~\cite{Coffman:1971:SD:356586.356588}. 


The entry point of the program implementing the Dining Philosophers example is the function {\verb+'make+} in the class {\verb+APPLICATION+}, which asks the two philosophers {\verb+p1+} and {\verb+p2+} to adopt a wrong eating strategy as above, possibly leading to a deadlock situation.
The flag enabling the deadlock analysis as in~(\ref{eq:Maude-deadlock}) is set to ``on''.
This information is specified using the instruction {\verb+settings('APPLICATION, 'make, deadlock-on)+}.

Unfortunately, none of the executions of the Dining Philosophers scenario by simply invoking the Maude {\verb+rewrite+} command lead to a deadlock situation. Each of our tests displayed the output:
\vspace{20pt}
{
\begin{verbatim}
                     \||||||||||||||||||/
                   --- Welcome to Maude ---
                     /||||||||||||||||||\
            Maude 2.6 built: Mar 31 2011 23:36:02
            Copyright 1997-2010 SRI International
			
rewrite in SYSTEM :
[...]  settings('APPLICATION, 'make, deadlock-on) 

|- {0}proc(0) :: nil | {0}proc(1) :: nil | {0}proc(3) :: nil
   {0}proc(5) :: nil | {0}proc(7) :: nil
   {0}proc(9) :: nil | {0}proc(11) :: nil
\end{verbatim}
}
\noindent
consisting of a list of processors (including the handlers of both the philosophers and the forks) with empty call stacks ({\verb+:: nil+}). This indicates that every time, the two philosophers proceeded by lifting their forks simultaneously, hence no deadlock was possible.

One possible workaround is to use predefined strategies~\cite{Marti-OlietMV05} in order to guide the rewriting of the Maude rules formalizing SCOOP towards a $\langle deadlock \rangle$ system configuration.
An example of applying such a strategy for the Dining Philosophers case can be tested by running the command:\\
{\verb+> maude SCOOP.maude ..\examples\dining-philosophers-strategy.maude+}

The command
{\verb+srew [...] using init ; parallelism{lock} ; [...] ; deadlock-on+} forces the execution of a {\verb+pick_in_turn+} approach as above.
This determines Maude to first trigger the rule {\verb+[init]+} in the SCOOP formalization in~\cite{DBLP:conf/acsd/MorandiSNM13}. This makes all the required initializations of the \emph{bootstrap} processor. Then, one of the processors that managed to \emph{lock} the necessary resources is (``randomly'') enabled to proceed to the asynchronous execution of its instruction stack, according to the strategy  {\verb+parallelism{lock}+} . The last step of the strategy calls the rule  {\verb+[deadlock-on]+} implementing the Coffman deadlock detection as in~(\ref{eq:Maude-deadlock}).

This time the guided rewriting leads, indeed, to one solution identifying a deadlock.
The relevant parts of the corresponding Maude output are as follows:
{
\begin{verbatim}
                     \||||||||||||||||||/
                   --- Welcome to Maude ---
                     /||||||||||||||||||\
            Maude 2.6 built: Mar 31 2011 23:36:02
            Copyright 1997-2010 SRI International

srewrite in SYSTEM :
[...] settings('APPLICATION, 'make, deadlock-on) 
using init ; parallelism{lock} ; [...] ; deadlock-on .

Solution 1
rewrites: 479677 in 2674887330ms cpu (7379ms real) (0 rewrites/second)
result Configuration: deadlock

No more solutions.
rewrites: 479677 in 2674887330ms cpu (7453ms real) (0 rewrites/second)
\end{verbatim}
}

Nevertheless, such an approach requires lots of ingeniousness (our strategy has more than 300 rules!) and, moreover, is not automated.

\subsection{Model-checking deadlocks}
\label{sec:deadlocking-model-check}

In this section we provide an overview on our approach to model-checking deadlocks for SCOOP, using the LTL Maude model-checker~\cite{DBLP:journals/entcs/EkerMS02}.
As mentioned in the beginning of the current paper, the idea behind our work is to exploit the unifying flavor of the Maude executable semantics of SCOOP~\cite{DBLP:conf/acsd/MorandiSNM13}. The latter integrates both the formalization of the language and its concurrency mechanisms, thus enabling using the semantic framework for program analysis purposes, ``for free''.

One possible way to proceed is by simply running the Maude LTL model-checker via a search command. For model-checking deadlocks, one could invoke {\verb+search [b, d] subject =>* deadlock+}. In a Dining Philosophers setting, for instance, this corresponds to performing a breadth-first search starting with {\verb+subject+} -- the corresponding (intermediate Maude representation of the) SCOOP program -- to a {\verb+deadlock+} state, in zero or more steps of proof ({\verb+=>*+}). The optional arguments {\verb+b+} and {\verb+d+} provide an upper bound in the number of solutions to be found and, respectively, the maximum depth of the search.

Unfortunately, running the aforementioned search command led to the state explosion problem.
At a first look, the issue was caused by the size of the SCOOP formalization in~\cite{DBLP:conf/acsd/MorandiSNM13} which includes all the aspects of a real concurrency model.

As a first step, we reduced this formalization by eliminating the parts that are not relevant in the context of deadlocking; examples include the garbage collection and the exception handling mechanisms.

In addition, we provided a simplified, abstract semantics of SCOOP based on aliasing. This idea stems from the fact that SCOOP processors are known from object references, that may be aliased. Therefore, the SCOOP semantics can be simplified by retaining within the corresponding transition rules only the information important for aliasing. Consider, for instance, the rules (\ref{eq:if-then-else})--(\ref{eq:if-else}) specifying ``{\bf{if}}'' instructions in Section~\ref{sec:SCOOP}. The abstract transition rule omits the evaluation of the conditional and computes the aliasing information similarly to the semantics of {\bf then \ldots else \ldots end} in~(\ref{eq:def-then-else}), in Section~\ref{sec:alias-calc}. The abstraction collects the aliases resulted after the execution of both ``{\bf if}'' and ``{\bf else}'' branches:
\begin{equation}
\label{eq:if-then-else-abstr}
\dfrac{.}{
\begin{array}{cc}
\langle p \,::\,\ifthenelse{e}{St_1}{St_2}\,;\, St,\, \sigma, alias_{old} \rangle \rightarrow\\
\langle p\,::\, St,\, \sigma, alias_{new}\rangle
\end{array}
}
\end{equation}
Observe that the SCOOP system configurations in~(\ref{eq:tuple}) are enriched with a new component $alias\_$ consisting of a set of alias expressions. Above, $alias_{old}$ is the aliasing before the execution of the ``{\bf if}'' instruction, whereas, intuitively, $alias_{new}$ stands for $alias_{old} \exec St_1 \cup alias_{old} \exec St_2$.

As a second step, we analyzed the implementation in~\cite{DBLP:conf/acsd/MorandiSNM13} from a more engineering perspective, and identified a series of design decisions that either slowed down considerably the rewriting or made the search space grow unnecessarily large.

After running some experiments, we understood that the parallelism rule
\begin{equation}
\label{eq:par}
\dfrac{\langle p_1\,::\, St_1,\, \sigma\rangle \rightarrow \langle p_1'\,::\, St_1',\, \sigma'\rangle}{
\langle p_1\,::\, St_1 \mid  p_2\,::\, St_2, \, \sigma\rangle \rightarrow \langle p_1'\,::\, St_1' \mid p_2\,::\, St_2 ,\, \sigma'\rangle
}
\end{equation}
in~\cite{DBLP:conf/acsd/MorandiSNM13} was increasing the rewriting time. Though elegant from the formalization perspective, the usage of this rule was not efficient. Therefore, we eliminated it from the SCOOP semantics and made the remaining rules apply directly, by matching at top. For instance, the abstract rule~(\ref{eq:if-then-else-abstr}) formalizing ``{\bf if}'' instructions in the context of one processor $p$ becomes:
\begin{equation}
\label{eq:if-then-else-abstr-gen}
\dfrac{.}{
\begin{array}{cc}
\langle p \,::\,\ifthenelse{e}{St_1}{St_2}\,;\, St \mid aqs,\, \sigma, alias_{old} \rangle \rightarrow\\
\langle p\,::\, St \mid aqs,\, \sigma, alias_{new}\rangle
\end{array}
}
\end{equation}
for an arbitrary list $aqs$ of processors and their instruction stacks.
For Dining Philosophers, for example, this modification reduced the rewriting time from around 10s to less than 1s.

Recall that SCOOP processors communicate via channels. The implementation in~\cite{DBLP:conf/acsd/MorandiSNM13} creates \emph{fresh} channels (as in~(\ref{eq:if-then-else}), for instance) parameterized by natural indexes. This was inefficient for model-checking purposes, as the state space contained many identical states up-to channel naming.

The implementation of the above observations enabled us to successfully identify a deadlock situation in a Dining Philosophers scenario, by using the Maude LTL model-checker.
The new (reduced) formalization of SCOOP can be found at:\\
{\verb+https://dl.dropboxusercontent.com/u/1356725/SCOOP-SCP.zip+}\\
whereas the example considering two philosophers can be run by executing the command\\
{ \verb+> maude SCOOP.maude ..\examples\dining-philosophers-model-check.maude+}

A bounded search successfully identifies one deadlock:
{
\begin{verbatim}
                     \||||||||||||||||||/
                   --- Welcome to Maude ---
                     /||||||||||||||||||\
            Maude 2.6 built: Mar 31 2011 23:36:02
            Copyright 1997-2010 SRI International
                
search [1, 200] in SYSTEM : 
[...] settings('APPLICATION, 'make, deadlock-on) =>* deadlock .

Solution 1 (state 769167)
states: 769168
rewrites: 342475817 in 2674905322ms cpu (1226403ms real)
\end{verbatim}
}

The amount of time necessary for model-checking is still quite large (approximately 20mins for the example above). However, further improvements may be obtained by following the same recipe of collapsing semantically equivalent states, from the deadlocking perspective. A major source of redundancy is represented by the so-called \emph{regions} in~\cite{DBLP:conf/acsd/MorandiSNM13} that, intuitively, manage all the objects handled by the same processor. Their elimination from the SCOOP abstract state would enable the model-checker to make less identifications, therefore improving the overall time performance.

\section{Discussion}
\label{sec:discussion}

The focus of this paper is on building a toolbox for the analysis of SCOOP programs, with emphasis on an alias analyzer and a deadlock detector. The naturalness of our approach consists in exploiting the recent formalization of SCOOP in~\cite{DBLP:conf/acsd/MorandiSNM13}, that is executable and implemented in Maude~\cite{DBLP:conf/maude/2007}.
This provides a unifying framework that can be used not only to reason about the SCOOP model and its design as in~\cite{DBLP:conf/acsd/MorandiSNM13}, but also to analyze SCOOP programs via Maude rewriting and model-checking.

Of particular interest for the aliasing tool is the calculus introduced in~\cite{Meyer-aliasing-13}, which abstracts the aliasing information in terms of explicit access paths referred to as ``alias expressions''. We provide an extension of this calculus from finite alias relations to infinite ones corresponding to loops and recursive calls. Moreover, we devise an associated RL-based executable specification in the {\K} semantic framework~\cite{DBLP:journals/jlp/RosuS10}. In Theorem~\ref{th:dec-proc} we show that the RL-based machinery implements an algorithm that always terminates with a sound over-approximation of ``may aliasing'', in non-concurrent settings. This is achieved based on the sound (finitely representable) over-approximation of alias expressions in terms of regular expressions, as in Lemma~\ref{lm:reg-expr}. 
The integration of the alias calculus on top of the Maude formalization of SCOOP~\cite{DBLP:conf/acsd/MorandiSNM13} is straightforward, based on the aforementioned executable specification of the calculus. Executions of SCOOP programs can be simulated by simply exploiting the Maude rewriting capabilities, hence the computation of the corresponding aliasing information is immediate.

A similar technique exploiting regular behaviour of (non-concurrent) programs, in order to reason on ``may aliasing'', was previously introduced in~\cite{Asav-aliasing-k}. In short, the results in~\cite{Asav-aliasing-k} utilize {abstract} representations of programs in terms of finite pushdown systems, for which infinite execution paths have a regular structure (or are ``lasso shaped'')~\cite{DBLP:conf/concur/BouajjaniEM97}. Then, in the style of abstract interpretation~\cite{DBLP:journals/jlp/CousotC92}, the collecting semantics is applied over the (finite state) pushdown systems to obtain the alias analysis itself.
The main difference with the results in~\cite{Asav-aliasing-k} consists in how the abstract memory addresses corresponding to pointer variables are represented. In~\cite{Asav-aliasing-k} these  range over a finite set of natural numbers. In this paper we consider alias expressions build according to the calculus in~\cite{Meyer-aliasing-13}, based on program constructs.
The work in~\cite{Asav-aliasing-k} also proposes an implementation of pushdown systems in the {\K}-framework.

We agree that it could be worth presenting our analysis as an abstract interpretation (AI)~\cite{DBLP:journals/jlp/CousotC92}. A modeling exploiting the machinery of AI (based on abstract domains, abstraction and concretization functions, Galois connections, fixed-points, {\emph etc.}) is an interesting topic left for future investigation.

An immediate direction for future work is to identify interesting (industrial) case studies to be analyzed using the framework developed in this paper.
We are also interested in devising heuristics comparing the efficiency and the precision ({\it e.g.}, the number of false positives introduced by the alias approximations) between our approach and other aliasing techniques.

Another research direction is to derive alias-based abstractions for analyzing concurrent programs. We foresee possible connections with the work in~\cite{DBLP:conf/concur/HoareMSW09} on {concurrent Kleene algebra} formalizing choice, iteration, sequential and concurrent composition of programs. The corresponding definitions exploit abstractions of programs in terms of traces of events that can depend on each other. Thus, obvious challenges in this respect include: (i) defining notions of dependence for all the program constructs in this paper, (ii) relating the concurrent Kleene operators to the semantics of the SCOOP concurrency model and (iii) checking whether fixed-points approximating the aliasing information can be identified via fixed-point theorems.

Furthermore, it would be worth investigating whether the graph-based model of alias relations introduced in~\cite{Meyer-aliasing-13} can be exploited in order to derive finite {\K} specifications of the extended alias calculus. In case of a positive answer, the general aim is to study whether this type of representation increases the speed of the reasoning mechanism, and why not -- its accuracy. With the same purpose, we refer to a possible integration with the technique in~\cite{DBLP:conf/pldi/ChaseWZ90} that handles point-to graphs via a stack-based algorithm for fixed-point computations.

Related to deadlock detection, the second contribution of this paper, we provided a formalization based on sets of acquired resources and sets of handlers processors wait to lock in their attempt to execute asynchronously. This definition corresponds to Coffman deadlocks~\cite{Coffman:1971:SD:356586.356588} in the context of SCOOP, occurring whenever there is a set of processors reserving each other circularly. We introduced the equivalent SCOOP semantic rule and discussed the results of using Maude in order to analyze deadlocks in the context of a Dining Philosophers scenario. On the one hand, the SCOOP semantics in~\cite{DBLP:conf/acsd/MorandiSNM13} is very large, as it incorporates all the aspects of a real concurrency model. On the other hand, the formalization in~\cite{DBLP:conf/acsd/MorandiSNM13} was tailored for reasoning about the SCOOP model, and not necessarily about SCOOP applications, hence it includes design decisions (such as index-based naming of communication channels) that make the state space grow unacceptably large for model-checking purposes. As a workaround for the model-checking issue, we presented the idea behind building an abstract semantics of SCOOP based on aliases, together with a series of implementation improvements that eventually enabled the Maude LTL model-checker to correctly identify deadlocks. A survey on abstracting techniques on top of Maude executable semantics is provided in~\cite{DBLP:conf/fct/MeseguerR11}.

The literature on using static analysis~\cite{Landi:1992:USA:161494.161501} and abstracting techniques for (related) concurrency models is considerable. We refer, for instance, to the recent work in~\cite{DBLP:conf/ifm/GiachinoGLLW13} that introduces a framework for detecting deadlocks by identifying circular dependencies in the (finite state) model of so-called contracts that abstract methods in an {OO}-language. 
Nevertheless, the integration of a deadlock analyzer in SCOOP on top of Maude is an orthogonal approach that aims at constructing a RL-based toolbox for SCOOP programs laying over the same semantic framework.

In~\cite{Heussner-PCM15a} SCOOP programs are verified for deadlocks and other behavioral properties using GROOVE~\cite{Ghamarian:2012:MAU:2150576.2150577}. The work in~\cite{Heussner-PCM15a} proposes a redefinition of the most common features of the SCOOP semantics based on graph transformation systems (GPSs). This is a bottom-up approach, as it aims at redefining the SCOOP semantics from scratch via GPSs, orthogonal to our rather top-down strategy of narrowing the original semantics proposed in~\cite{DBLP:conf/acsd/MorandiSNM13}.

As a clear direction for future work we consider designing and analyzing deadlock situations for more SCOOP applications. Based on the experience so far, this would help better understand and observe the SCOOP state space, thus providing hints for further improvements in the context of model-checking, for instance. As we could already see, major advances in this regard are obtained when semantically equivalent states are identified and collapsed within the same equivalence classes. This was the case of the indexed-based communication channels in Section~\ref{sec:deadlocking-model-check}. Similar redundant states are introduced by the so-called ``regions'' in SCOOP ``administrating'' objects handled by the same processor. We foresee their elimination would speed-up the LTL model-checker.

\paragraph{Acknowledgements}
{
We are grateful for valuable comments to the anonymous reviewers of FTSCS'14,
M\u ariuca As\u avoae,  Alexander Kogtenkov, Jos\'e Meseguer, Benjamin Morandi and Sergey Velder.

The research leading to these results has received funding from the
European Research Council under the European Union's Seventh Framework
Programme (FP7/2007-2013) / ERC Grant agreement no. 291389.
}
\vspace{20pt}




 \bibliographystyle{elsarticle-num} 
 \bibliography{circ}

\begin{thebibliography}{10}
\expandafter\ifx\csname url\endcsname\relax
  \def\url#1{\texttt{#1}}\fi
\expandafter\ifx\csname urlprefix\endcsname\relax\def\urlprefix{URL }\fi
\expandafter\ifx\csname href\endcsname\relax
  \def\href#1#2{#2} \def\path#1{#1}\fi

\bibitem{DBLP:conf/acsd/MorandiSNM13}
B.~Morandi, M.~Schill, S.~Nanz, B.~Meyer, Prototyping a concurrency model, in:
  ACSD, 2013, pp. 170--179.

\bibitem{DBLP:books/ph/Meyer91}
B.~Meyer, Eiffel: The Language, Prentice-Hall, 1991.

\bibitem{DBLP:conf/iros/RusakovSM14}
A.~Rusakov, J.~Shin, B.~Meyer,
  \href{http://dx.doi.org/10.1109/IROS.2014.6942763}{Simple concurrency for
  robotics with the {R}oboscoop framework}, in: 2014 {IEEE/RSJ} International
  Conference on Intelligent Robots and Systems, Chicago, IL, USA, September
  14-18, 2014, 2014, pp. 1563--1569.
\newblock \href {http://dx.doi.org/10.1109/IROS.2014.6942763}
  {\path{doi:10.1109/IROS.2014.6942763}}.
\newline\urlprefix\url{http://dx.doi.org/10.1109/IROS.2014.6942763}

\bibitem{DBLP:conf/fct/MeseguerR11}
J.~Meseguer, G.~Rosu, The rewriting logic semantics project: {A} progress
  report., in: Fundamentals of Computation Theory - 18th International
  Symposium, {FCT} 2011, Oslo, Norway, August 22-25, 2011. Proceedings, 2011,
  pp. 1--37.

\bibitem{DBLP:conf/maude/2007}
M.~Clavel, F.~Dur{\'a}n, S.~Eker, P.~Lincoln, N.~Mart\'{\i}-Oliet, J.~Meseguer,
  C.~L. Talcott (Eds.), All About Maude - A High-Performance Logical Framework,
  How to Specify, Program and Verify Systems in Rewriting Logic, Vol. 4350 of
  Lecture Notes in Computer Science, Springer, 2007.

\bibitem{Landi:1991:PAP:99583.99599}
W.~Landi, B.~G. Ryder,
  \href{http://doi.acm.org/10.1145/99583.99599}{Pointer-induced aliasing: A
  problem classification}, in: Proceedings of the 18th ACM SIGPLAN-SIGACT
  Symposium on Principles of Programming Languages, POPL '91, ACM, New York,
  NY, USA, 1991, pp. 93--103.
\newblock \href {http://dx.doi.org/10.1145/99583.99599}
  {\path{doi:10.1145/99583.99599}}.
\newline\urlprefix\url{http://doi.acm.org/10.1145/99583.99599}

\bibitem{Landi:1992:USA:161494.161501}
W.~Landi, \href{http://doi.acm.org/10.1145/161494.161501}{Undecidability of
  static analysis}, ACM Lett. Program. Lang. Syst. 1~(4) (1992) 323--337.
\newblock \href {http://dx.doi.org/10.1145/161494.161501}
  {\path{doi:10.1145/161494.161501}}.
\newline\urlprefix\url{http://doi.acm.org/10.1145/161494.161501}

\bibitem{Myers:1981:PID:567532.567556}
E.~M. Myers, \href{http://doi.acm.org/10.1145/567532.567556}{A precise
  inter-procedural data flow algorithm}, in: Proceedings of the 8th ACM
  SIGPLAN-SIGACT Symposium on Principles of Programming Languages, POPL '81,
  ACM, New York, NY, USA, 1981, pp. 219--230.
\newblock \href {http://dx.doi.org/10.1145/567532.567556}
  {\path{doi:10.1145/567532.567556}}.
\newline\urlprefix\url{http://doi.acm.org/10.1145/567532.567556}

\bibitem{Hind:1999:IPA:325478.325519}
M.~Hind, M.~Burke, P.~Carini, J.-D. Choi,
  \href{http://doi.acm.org/10.1145/325478.325519}{Interprocedural pointer alias
  analysis}, ACM Trans. Program. Lang. Syst. 21~(4) (1999) 848--894.
\newblock \href {http://dx.doi.org/10.1145/325478.325519}
  {\path{doi:10.1145/325478.325519}}.
\newline\urlprefix\url{http://doi.acm.org/10.1145/325478.325519}

\bibitem{Diwan:1998:TAA:277652.277670}
A.~Diwan, K.~S. McKinley, J.~E.~B. Moss,
  \href{http://doi.acm.org/10.1145/277652.277670}{Type-based alias analysis},
  SIGPLAN Not. 33~(5) (1998) 106--117.
\newblock \href {http://dx.doi.org/10.1145/277652.277670}
  {\path{doi:10.1145/277652.277670}}.
\newline\urlprefix\url{http://doi.acm.org/10.1145/277652.277670}

\bibitem{Burke-flow-ins}
M.~Burke, P.~Carini, J.-D. Choi, M.~Hind,
  \href{http://dx.doi.org/10.1007/BFb0025882}{Flow-insensitive interprocedural
  alias analysis in the presence of pointers}, in: K.~Pingali, U.~Banerjee,
  D.~Gelernter, A.~Nicolau, D.~Padua (Eds.), Languages and Compilers for
  Parallel Computing, Vol. 892 of Lecture Notes in Computer Science, Springer
  Berlin Heidelberg, 1995, pp. 234--250.
\newblock \href {http://dx.doi.org/10.1007/BFb0025882}
  {\path{doi:10.1007/BFb0025882}}.
\newline\urlprefix\url{http://dx.doi.org/10.1007/BFb0025882}

\bibitem{Choi:1993:EFI:158511.158639}
J.-D. Choi, M.~Burke, P.~Carini,
  \href{http://doi.acm.org/10.1145/158511.158639}{Efficient flow-sensitive
  interprocedural computation of pointer-induced aliases and side effects}, in:
  Proceedings of the 20th ACM SIGPLAN-SIGACT Symposium on Principles of
  Programming Languages, POPL '93, ACM, New York, NY, USA, 1993, pp. 232--245.
\newblock \href {http://dx.doi.org/10.1145/158511.158639}
  {\path{doi:10.1145/158511.158639}}.
\newline\urlprefix\url{http://doi.acm.org/10.1145/158511.158639}

\bibitem{Emami:1994:CIP:178243.178264}
M.~Emami, R.~Ghiya, L.~J. Hendren,
  \href{http://doi.acm.org/10.1145/178243.178264}{Context-sensitive
  interprocedural points-to analysis in the presence of function pointers}, in:
  Proceedings of the ACM SIGPLAN 1994 Conference on Programming Language Design
  and Implementation, PLDI '94, ACM, New York, NY, USA, 1994, pp. 242--256.
\newblock \href {http://dx.doi.org/10.1145/178243.178264}
  {\path{doi:10.1145/178243.178264}}.
\newline\urlprefix\url{http://doi.acm.org/10.1145/178243.178264}

\bibitem{Wilson:1995:ECP:207110.207111}
R.~P. Wilson, M.~S. Lam,
  \href{http://doi.acm.org/10.1145/207110.207111}{Efficient context-sensitive
  pointer analysis for {C} programs}, in: Proceedings of the ACM SIGPLAN 1995
  Conference on Programming Language Design and Implementation, PLDI '95, ACM,
  New York, NY, USA, 1995, pp. 1--12.
\newblock \href {http://dx.doi.org/10.1145/207110.207111}
  {\path{doi:10.1145/207110.207111}}.
\newline\urlprefix\url{http://doi.acm.org/10.1145/207110.207111}

\bibitem{Mine:2006:FVA:1134650.1134659}
A.~Min{\'e}, \href{http://doi.acm.org/10.1145/1134650.1134659}{Field-sensitive
  value analysis of embedded {C} programs with union types and pointer
  arithmetics}, in: Proceedings of the 2006 ACM SIGPLAN/SIGBED Conference on
  Language, Compilers, and Tool Support for Embedded Systems, LCTES '06, ACM,
  New York, NY, USA, 2006, pp. 54--63.
\newblock \href {http://dx.doi.org/10.1145/1134650.1134659}
  {\path{doi:10.1145/1134650.1134659}}.
\newline\urlprefix\url{http://doi.acm.org/10.1145/1134650.1134659}

\bibitem{Albert:2009:FVA:1693345.1693376}
E.~Albert, P.~Arenas, S.~Genaim, G.~Puebla, Field-sensitive value analysis by
  field-insensitive analysis, in: Proceedings of the 2Nd World Congress on
  Formal Methods, FM '09, Springer-Verlag, Berlin, Heidelberg, 2009, pp.
  370--386.

\bibitem{DBLP:conf/paste/Hind01}
M.~Hind, Pointer analysis: haven't we solved this problem yet?, in: PASTE,
  2001, pp. 54--61.

\bibitem{Meyer-aliasing-13}
A.~Kogtenkov, B.~Meyer, S.~Velder, \href{http://arxiv.org/abs/1307.3189}{Alias
  and change calculi, applied to frame inference}, CoRR abs/1307.3189.
\newline\urlprefix\url{http://arxiv.org/abs/1307.3189}

\bibitem{DBLP:conf/pldi/LarusH88}
J.~R. Larus, P.~N. Hilfinger, Detecting conflicts between structure accesses,
  in: PLDI, 1988, pp. 21--34.

\bibitem{Shih:1990:SDD:896944}
C.~Shih, J.~A. Stankovic, Survey of deadlock detection in distributed
  concurrent programming environments and its application to real-time systems,
  Tech. rep., Amherst, MA, USA (1990).

\bibitem{Andrews:1982:ODP:800220.806694}
G.~R. Andrews, G.~M. Levin, On-the-fly deadlock prevention, in: Proceedings of
  the First ACM SIGACT-SIGOPS Symposium on Principles of Distributed Computing,
  PODC '82, ACM, New York, NY, USA, 1982, pp. 165--172.

\bibitem{Minoura:1982:DAR:322344.322351}
T.~Minoura, Deadlock avoidance revisited, J. ACM 29~(4) (1982) 1023--1048.

\bibitem{chandy1982distributed}
K.~Chandy, J.~Misra, L.~Haas, T.~U. A. A. D. O.~C. SCIENCES.,
  \href{http://books.google.ch/books?id=AH-3NwAACAAJ}{A Distributed Deadlock
  Detection Algorithm and Its Correctness Proof}, Defense Technical Information
  Center, 1982.
\newline\urlprefix\url{http://books.google.ch/books?id=AH-3NwAACAAJ}

\bibitem{badal1983deadlock}
D.~Badal, M.~Gehl, M.~Gehl, N.~P. S.~M. CA., N.~P.~S. (U.S.),
  \href{http://books.google.ch/books?id=JjrqOAAACAAJ}{On Deadlock Detection in
  Distributed Computing Systems}, Defense Technical Information Center, 1983.
\newline\urlprefix\url{http://books.google.ch/books?id=JjrqOAAACAAJ}

\bibitem{DBLP:journals/corr/Caltais14}
G.~Caltais, \href{http://arxiv.org/abs/1409.7509}{Expression-based aliasing for
  {OO}-languages}, CoRR abs/1409.7509, {accepted in FTSCS'14}.
\newline\urlprefix\url{http://arxiv.org/abs/1409.7509}

\bibitem{Rabin:1959:FAD:1661907.1661909}
M.~O. Rabin, D.~Scott, \href{http://dx.doi.org/10.1147/rd.32.0114}{Finite
  automata and their decision problems}, IBM J. Res. Dev. 3~(2) (1959)
  114--125.
\newblock \href {http://dx.doi.org/10.1147/rd.32.0114}
  {\path{doi:10.1147/rd.32.0114}}.
\newline\urlprefix\url{http://dx.doi.org/10.1147/rd.32.0114}

\bibitem{DBLP:journals/jlp/RosuS10}
G.~Rosu, T.~F. Serbanuta, K overview and {SIMPLE} case study, in: Proceedings
  of International K Workshop (K'11), ENTCS, Elsevier, 2013, {To appear.}

\bibitem{Kleene56}
S.~C. Kleene, Representation of events in nerve nets and finite automata, in:
  C.~Shannon, J.~McCarthy (Eds.), Automata Studies, Princeton University Press,
  Princeton, NJ, 1956, pp. 3--41.

\bibitem{DBLP:conf/wrla/SerbanutaR10}
T.-F. Serbanuta, G.~Rosu, K-{M}aude: A rewriting based tool for semantics of
  programming languages, in: WRLA, 2010, pp. 104--122.

\bibitem{nienaltowski2007practical}
P.~Nienaltowski, \href{http://books.google.ch/books?id=ZDcEkgAACAAJ}{Practical
  Framework for Contract-based Concurrent Object-oriented Programming}, ETH,
  2007.
\newline\urlprefix\url{http://books.google.ch/books?id=ZDcEkgAACAAJ}

\bibitem{Coffman:1971:SD:356586.356588}
E.~G. Coffman, M.~Elphick, A.~Shoshani, System deadlocks, ACM Comput. Surv.
  3~(2) (1971) 67--78.

\bibitem{Marti-OlietMV05}
N.~Mart\'{\i}-Oliet, J.~Meseguer, A.~Verdejo, {Towards a Strategy Language for
  Maude}, Electronic Notes in Theoretical Computer Science 117 (2005) 417--441.

\bibitem{DBLP:journals/entcs/EkerMS02}
S.~Eker, J.~Meseguer, A.~Sridharanarayanan, The maude {LTL} model checker,
  Electr. Notes Theor. Comput. Sci. 71 (2002) 162--187.

\bibitem{Asav-aliasing-k}
I.~M. Asavoae, \href{http://dx.doi.org/10.1016/j.entcs.2014.05.005}{Abstract
  semantics for alias analysis in {K}}, Electr. Notes Theor. Comput. Sci. 304
  (2014) 97--110.
\newblock \href {http://dx.doi.org/10.1016/j.entcs.2014.05.005}
  {\path{doi:10.1016/j.entcs.2014.05.005}}.
\newline\urlprefix\url{http://dx.doi.org/10.1016/j.entcs.2014.05.005}

\bibitem{DBLP:conf/concur/BouajjaniEM97}
A.~Bouajjani, J.~Esparza, O.~Maler, Reachability analysis of pushdown automata:
  Application to model-checking, in: CONCUR, 1997, pp. 135--150.

\bibitem{DBLP:journals/jlp/CousotC92}
P.~Cousot, R.~Cousot, Abstract interpretation and application to logic
  programs, J. Log. Program. 13~(2{\&}3) (1992) 103--179.

\bibitem{DBLP:conf/concur/HoareMSW09}
C.~A.~R. Hoare, B.~M{\"{o}}ller, G.~Struth, I.~Wehrman, Concurrent {K}leene
  algebra, in: {CONCUR} 2009 - Concurrency Theory, 20th International
  Conference, {CONCUR} 2009, Bologna, Italy, September 1-4, 2009. Proceedings,
  2009, pp. 399--414.

\bibitem{DBLP:conf/pldi/ChaseWZ90}
D.~R. Chase, M.~N. Wegman, F.~K. Zadeck, Analysis of pointers and structures,
  in: PLDI, 1990, pp. 296--310.

\bibitem{DBLP:conf/ifm/GiachinoGLLW13}
E.~Giachino, C.~A. Grazia, C.~Laneve, M.~Lienhardt, P.~Y.~H. Wong, Deadlock
  analysis of concurrent objects: Theory and practice, in: Integrated Formal
  Methods, 10th International Conference, {IFM} 2013, Turku, Finland, June
  10-14, 2013. Proceedings, 2013, pp. 394--411.

\bibitem{Heussner-PCM15a}
A.~Heu{\ss}ner, C.~M. Poskitt, C.~Corrodi, B.~Morandi, Towards practical
  graph-based verification for an object-oriented concurrency model, in: Proc.\
  Graphs as Models (GaM 2015), Vol. 181 of Electronic Proceedings in
  Theoretical Computer Science, 2015, pp. 32--47.

\bibitem{Ghamarian:2012:MAU:2150576.2150577}
A.~H. Ghamarian, M.~de~Mol, A.~Rensink, E.~Zambon, M.~Zimakova, Modelling and
  analysis using groove, Int. J. Softw. Tools Technol. Transf. 14~(1) (2012)
  15--40.

\end{thebibliography}


%
%
%
\end{document}